\documentclass[12pt]{article}
\usepackage[utf8]{inputenc}

\usepackage{makeidx}         
\usepackage{graphicx}        
\usepackage{multicol}        
\usepackage{multirow}
\usepackage[bottom]{footmisc}
\usepackage{geometry} 
\usepackage{amsmath, amssymb, amsthm}
\usepackage{psfrag,epsf}
\usepackage{enumerate}
\usepackage{url} 
\usepackage{hyperref}
\usepackage{booktabs}
\usepackage{xcolor}
\usepackage{natbib}
\usepackage[italian,english]{babel}
\usepackage{authblk}
\usepackage{changes}
\usepackage{subcaption}
\usepackage{placeins}
\usepackage{pdfpages}

\hypersetup{
			colorlinks=true,
			linkcolor=black,
            linktoc=page,
			anchorcolor=black,
			citecolor=blue,
			urlcolor=black,
}

\def \bs{\mathbf}

\newtheorem{proposition}{Proposition}

\theoremstyle{definition}
\newtheorem{definition}{Definition}

\def\0{\mbox{\bf{0}}}
\def\bs{\mathbf{s}}\def\be{\mathbf{e}}

\usepackage[toc,page]{appendix}

\def \be{\begin{equation}}
\def \ee{\end{equation}}
\def \ber{\begin{eqnarray}}
\def \eer{\end{eqnarray}}
\def \berr{\begin{eqnarray*}}
\def \eerr{\end{eqnarray*}}
\def \bqmatrix{\begin{bmatrix}}
\def \eqmatrix{\end{bmatrix}}

\def \be{\begin{equation}}
\def \ee{\end{equation}}
\def \ber{\begin{eqnarray}}
\def \eer{\end{eqnarray}}
\def \berr{\begin{eqnarray*}}
\def \eerr{\end{eqnarray*}}
\def \bamatrix{\begin{pmatrix}}
\def \eamatrix{\end{pmatrix}}
\def \bqmatrix{\begin{bmatrix}}
\def \eqmatrix{\end{bmatrix}}

\def \bs{\boldsymbol}

\def \qmo{``}

\def \qmcsp{'' }

\begin{document}
\title{\bf }
\author[1]{Beatrice Foroni}
\affil[1]{Department of Management and Economics, University of Pisa}
\author[2]{Luca Merlo}
\affil[2]{Department of Human Sciences,  Link Campus University of Rome}
\author[3]{Lea Petrella}
\affil[3]{MEMOTEF Department, Sapienza University of Rome}

\title{
Hidden Markov graphical models with state-dependent generalized hyperbolic distributions}

\maketitle

\begin{abstract}
In this paper we develop a novel hidden Markov graphical model to investigate time-varying interconnectedness between different financial markets. To identify conditional correlation structures under varying market conditions and accommodate stylized facts embedded in financial time series, we rely upon the generalized hyperbolic family of distributions with time-dependent parameters evolving according to a latent Markov chain. We exploit its location-scale mixture representation to build a penalized EM algorithm for estimating the state-specific sparse precision matrices by means of an $L_1$ penalty. The proposed approach leads to regime-specific conditional correlation graphs that allow us to identify different degrees of network connectivity of returns over time. The methodology's effectiveness is validated through simulation exercises under different scenarios. 
 In the empirical analysis we apply our model to daily returns of a large set of market indexes, cryptocurrencies and commodity futures over the period 2017-2023.
\end{abstract}
\noindent%
\vfill


\section{Introduction}\label{sec:intro}
\noindent{The financial system} is a complex, dynamic and interconnected world. Observing the extreme financial integration in the recent global crisis, it was soon noted the crucial importance to identify how the impact of stress events can spread across the whole financial global system. In the last years, the development of statistical techniques to accurately quantify and investigate the interrelations among financial institutions has been at the centre of the attention not only of investors and fund managers, but also of regulators for early identification of systemic risk and proactively engaging preventing measures to control financial stability (\citealt{silva2017analysis, silva2018bank, giudici2018corisk, brunetti2019interconnectedness}). For this reason, network science has emerged as a useful tool for describing the propagation of systemic risk, where the interconnectedness between markets is represented by a graph whose nodes stand for companies, commodities, institutions, for instance, and edges represent their interactions.
 In this context, Gaussian Graphical Models (GGMs) have received an enormous attention because they provide a simple method to model the pair-wise conditional correlations of a collection of stochastic variables. As it is well known, for normally distributed data, the underlying conditional dependence structure is completely characterized by the inverse of the covariance matrix, also known as precision or concentration matrix, of the corresponding GGM (see \citealt{lauritzen1996graphical} for a general background). Thus parameter estimation in GGMs is a fundamental issue for the identification of the zero entries in the concentration matrix. 
When dealing with large dimensional problems, we are interested in identifying only the variables that exhibit the most relevant and strongest connections. Among the several graphical methods proposed in literature (\citealt{dempster1972covariance, drton2004model, banerjee2008model, drton2008sinful, cai2011constrained, liu2012high, liu2017tiger}), there exists the popular and computationally efficient \textit{Graphical Lasso} (\textit{glasso}) algorithm of \cite{friedman2008sparse}, which maximizes the likelihood of the model penalized by the $L_1$-norm of the elements of the precision matrix. \\
Moreover, financial time series are characterized by well-known stylized facts like fat tails, leptokurtosis and deviations from normality, that the analysis needs to properly take into account. Especially in the highly turbulent cryptocurrency market, risk managers and regulators are increasingly interested in determining whether, and how, the temporal evolution and volatility clustering of returns can be influenced by hidden variables, e.g., the state of the market, during tranquil and crisis periods.
 In this context, Hidden Markov Models (HMMs, see \citealt{macdonald1997hidden, zucchini2016hidden}) have been successfully employed in the analysis of financial time series data, with applications to asset allocation and stock returns as discussed in \cite{mergner2008time, de2013dynamic, nystrup2017long, maruotti2019hidden, pennoni2022exploring} and \cite{foroni2023expectile, foroni2023quantile}. 
In the context of undirected graphs, estimation of graphical models in HMMs has been addressed by \cite{stadler2013penalized} using multivariate Gaussian emission distributions with sparse precision matrices which can be interpreted as state-specific conditional independence graphs. More recently, \cite{bianchi2019modeling} introduced a Markov switching graphical seemingly unrelated regression model to investigate time-varying systemic risk based on a range of multi-factor asset pricing models. 
Nevertheless, \qmo non-standard\qmcsp features of returns cannot be accommodated by standard models such as those based on normality assumptions. Unfortunately, the literature regarding non-Gaussian graphical models is fairly limited. Potential modeling strategies could rely on semiparametric Gaussian copula models (\citealt{liu2012high, xue2012regularized}), such as the Nonparanormal model of \cite{liu2009nonparanormal}, or power transformations of the data. 
 Alternatively, \cite{finegold2011robust} have introduced a robust graphical model based on the multivariate t distribution called \textit{tlasso}. 

In this paper, we contribute to the existing literature by introducing a sparse hidden Markov graphical model to investigate time-varying conditional correlation structures in multivariate time series data, without assuming normally distributed returns. To build our network model, we consider multivariate symmetric generalized hyperbolic (GH, \citealt{mcneil2015quantitative}) distributions with time-dependent parameters evolving according to a discrete, homogeneous latent Markov chain. Within the financial literature, this family of densities has garnered significant attention for describing pertinent features of the distribution of returns (\citealt{chen2008nonparametric, necula2009modeling, ignatieva2015estimating}) but also for its considerable flexibility in modeling financial data (\citealt{konlack2014comparison, zhang2019generalised}) which includes the multivariate Normal, t, Laplace and several others as particular cases (\citealt{mcneil2015quantitative, browne2015mixture}). 
Following \cite{finegold2011robust}, we demonstrate that, conditionally on each latent state, the inverse of the state-specific scale matrix of the multivariate GH completely characterizes the conditional correlation structure among the random variables. 
 To induce sparsity in the inverse of the scale matrices, i.e., the precision matrices, and identify whether two nodes are connected by an edge, we exploit the Gaussian location-scale mixture representation of the GH family to build a suitable penalized expectation-conditional maximization either (ECME) algorithm. This enables us to include in our estimation procedure the \textit{glasso} approach of \cite{friedman2008sparse} accounting for an $L_1$ penalty on the off-diagonal elements of the precision matrices in the M-step of the algorithm. We call this method \textit{hidden Markov generalized hyperbolic graphical model} (HMGHGM). 
 Within this scheme, our modeling framework allows us to construct a set of regime-specific graphs whose set of edges is determined by the non-zero elements of the estimated state-specific precision matrices. As opposed to GGMs, the proposed methodology has several advantages. Firstly, it enables us to estimate levels of network connectivity among asset returns in different market phases corresponding to different states of the latent process. Secondly, we don't rely on the restrictive assumption of normally distributed data, but instead each state-dependent GH distribution has shape parameters that are free to vary within the GH family to provide the best fit to the data.
 Using simulation exercises, we validate the ability of our method to correctly (i) recover the true values of the parameters under different states of the Markov chain, (ii) identify the true HMM clustering partition and (iii) retrieve the graphical model by recovering the true edges of the graphs.
 
Empirically, we analyze daily returns of a large set of financial assets, including the most important world stock market indices, commodity futures and the largest cryptocurrencies by market capitalization. In the last 10 years, the emergence of the cryptocurrency market has increasingly attracted the attention of market participants. The shortage of safe assets after the global financial crisis of $2008$ has raised several concerns among investors and researchers on whether digital currencies can offer hedging and safe-haven abilities for equity investments (\citealt{bouri2020cryptocurrencies}).
In the current literature, there are very few studies examining how traditional asset classes, (i.e., stocks, bonds and exchange rates), commodities and different cryptocurrencies interact with each other. The existing approaches have mainly focused on the relationships between a limited collection of cryptocurrencies and assets (\citealt{baur2018bitcoin, corbet2018exploring, ji2018network, bouri2020bitcoin, chen2020lead, giudici2021crypto}) or relied on conditional means to identify correlation structures (\citealt{bouri2017hedge}), which cannot provide a complete picture of the dependencies and the transmission path of risks between markets. Here, we implement the proposed HMGHGM to analyze the conditional correlation structure among the cryptocurrency, commodity and stock market sectors from 2017 to 2023 and evaluate how it may vary when considering different volatility clusters. During the period considered, there were significant episodes of prices fluctuations and instability, such as the explosion in cryptocurrencies at the start of 2018, the COVID-19 pandemic erupted in 2020, which have triggered unexpected levels of uncertainty and high volatility, and the EU sanctions to Russia dictated by the Ukrainian war. In this way, we are able to study how the degree of correlation changes as
different hidden volatility states are considered.

  To the best of our knowledge, this is the first attempt to build a non-gaussian hidden Markov graphical model for estimating time-varying cross-market conditional association structures of financial returns.\\ 
  
  The rest of the paper is organized as follows. In Section \ref{sec:model}, we briefly review the GH distribution and formally introduce the HMGHGM. Section \ref{sec:Est} proposes the ECME-based maximum likelihood approach and the related penalized algorithm for sparse estimation of the state-specific precision matrices. In Section \ref{sec:sim} we provide simulation results, while the empirical application is presented in Section \ref{sec:app}. Section \ref{sec:concl} summarizes our conclusions.

\section{Methodology}\label{sec:model}
In this section we introduce the hidden Markov generalized hyperbolic graphical model (HMGHGM). Before describing the model, we briefly revise the symmetric multivariate GH distribution, its location-scale mixture representation and its limiting cases. 
 Subsequently, we show how it is possible to build sparse state-specific graphical models for characterizing time-varying conditional correlation relations among variables.

\subsection{The generalized hyperbolic distribution and its special cases}\label{sub:GH}

Formally, let $\mathbf{Y}_{t} = [Y_{t}^{(1)}, \dots, Y_{t}^{(d)}]'$ denote a continuous $d$-dimensional random vector for $t = 1, \dots, T$. The joint probability density function of $\bs Y_t$ following the (symmetric) GH distribution can be written as
\begin{equation}\label{GH_dens}
	f_{\bs{Y}_t} (\bs y_{t}; \bs \mu, \bs \Sigma, \lambda, \chi, \psi) = \frac{1}{(2\pi)^{d/2}|\bs \Sigma|^{1/2}K_{\lambda}(\sqrt
	{\psi \chi})} \left[ \frac{\chi + \delta(\bs y_{t}; \bs \mu, \bs \Sigma)}{\psi} \right]^{\frac{\lambda - d/2}{2}} K_{\lambda - \frac{d}{2}} \left( \sqrt{\left[ \chi + \delta(\bs y_{t}; \bs \mu, \bs \Sigma) \right] \psi} \right)
\end{equation}
where $\boldsymbol{\mu} \in \mathbb{R}^d$ is the location parameter, $\bs { \Sigma}$ is a $d \times d$ positive definite and symmetric scale matrix, such that $|\bs \Sigma| = 1$ for identifiability purposes (see \cite{mcneil2015quantitative} for details), $\lambda \in \mathbb{R}$ is the index parameter, $\chi > 0$ and $\psi >0$ are concentration parameters, $\delta(\bs y_{t}; \bs \mu, \bs \Sigma) = (\bs y_{t} - \bs \mu)' \bs \Sigma^{-1} (\bs y_{t} - \bs \mu)$ is the squared Mahalanobis distance between $\bs y_{t}$ and $\bs \mu$ with scale matrix $\bs \Sigma$ and finally $K_{\lambda - \frac{d}{2}}(\cdot)$ denotes the modified Bessel function of the third kind of order $\lambda - \frac{d}{2}$. We adopt the compact notation $\bs{Y}_t \sim \mathcal{GH}_d(\bs \mu, \bs \Sigma, \lambda, \chi, \psi)$.
One of the key benefits of the GH distribution is that, using \eqref{GH_dens}, $\bs{Y}_t$ admits the following location-scale mixture representation:
\begin{equation}\label{mixtureGH}
\mathbf{Y}_t = \boldsymbol{\mu} + \sqrt{W_t} \bs {\Sigma}^{1/2}\bs Z_t
\end{equation}
where $\bs Z_t \sim {\cal N}_d(\bs 0_d, \bs I_d)$ denotes a $d$-variate standard Normal distribution whose covariance matrix is the identity matrix and $W_t$ has a generalized inverse Gaussian (GIG) distribution, $W_t \sim \mathcal{GIG}(\lambda, \chi, \psi)$, with $\bs Z_t$ being independent of $W_t$.
From \eqref{mixtureGH}, we can refer to the following hierarchical representation of $\bs{Y}_t \sim \mathcal{GH}_d(\bs \mu, \bs \Sigma, \lambda, \chi, \psi)$:
\begin{equation}
\begin{aligned}\label{GHrep2}
 	 W_t \sim & \ \mathcal{GIG}(\lambda, \chi, \psi), \\
 	  \bs{Y}_t|W_t = w_t \sim & \ {\cal N}_d(\bs \mu, w_t\bs \Sigma)
\end{aligned}
\end{equation}
which is useful for random data generation and for the implementation of our ECME algorithm.

The GH family encompasses several well-known, applied models for financial data by varying appropriately the values of the parameters of the GIG distribution, $(\lambda, \chi, \psi)$, including the multivariate Laplace, t and Normal distribution as presented in Figure \ref{fig:GHfam}. 
 
A complete taxonomy of all the models belonging to the class of GH densities can be found, for instance, in \cite{mcneil2015quantitative, browne2015mixture} and \cite{bagnato2023generalized}.

\begin{figure}[!h]
\centering
  \includegraphics[width=1\linewidth]{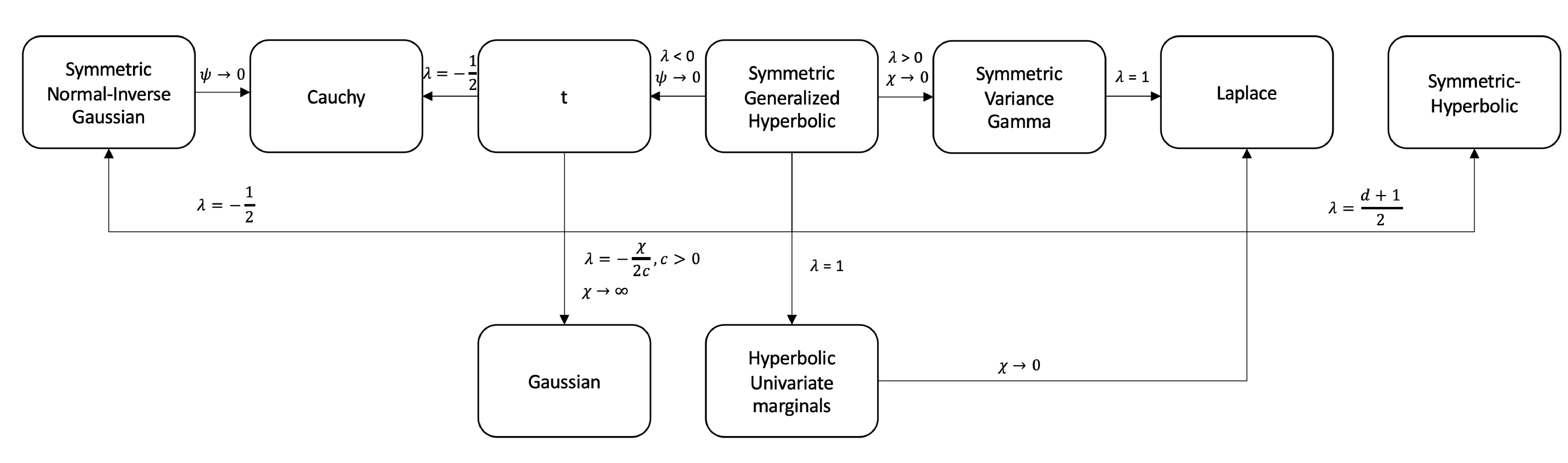}
\caption{Special and limiting cases of the symmetric GH distribution in terms of $\lambda, \chi$ and $\psi$.}
\label{fig:GHfam}
\end{figure}
\FloatBarrier

\subsection{Hidden Markov graphical models with state-specific GH densities}
In this section we describe the proposed hidden Markov graphical model with state-specific multivariate GH emission densities. 
\\
 Formally, let $\{ S_t \}_{t=1}^T$ be a latent, homogeneous, first-order Markov chain defined on the discrete state space $\{1,\dots,K\}$. Let $\pi_k = Pr(S_1=k)$ be the initial probability of state $k$, $k = 1,\dots,K$, and $\pi_{k\lvert j} = Pr(S_{t+1}=k \lvert S_t=j)$, with $\sum_{k=1}^K \pi_{k \lvert j} = 1$ and $\pi_{k\lvert j} \geq 0$, denote the transition probability between states $j$ and $k$, that is, the probability to visit state $k$ at time $t+1$ from state $j$ at time $t$, $j,k = 1,\dots,K$ and $t=1,\dots,T$. We collect the initial and transition probabilities in the $K$-dimensional vector $\bs \pi$ and in the $K \times K$ matrix $\bs \Pi$, respectively. 
\\
In our model we assume that the conditional distribution of $\bs Y_t$ given the state occupied by the latent process at time $t$ 
 corresponds to a GH distribution in \eqref{GH_dens} whose parameters depend on the value of the Markov chain $S_t$, namely $\bs Y_t \mid S_t = k \sim \mathcal{GH}_d(\bs \mu_k, \bs \Sigma_k, \lambda_k, \chi_k, \psi_k)$. We define with $\bs \Theta_k = \bs \Sigma_k^{-1}$ for $k=1,\dots,K$, the precision matrices used to build state-specific undirected graphs that conveys the conditional intercorrelation structure among the elements of $\bs Y_t$ given the latent process $S_t$. More precisely, supposing $S_t = k$, let $G_k = (V, E_k)$ be an undirected state-dependent graph where $V= \{1,\dots, d\}$ denotes the set of state-invariant nodes, such that each component of the random variable $\bs Y_t$ corresponds to a node in $V$, and $E_k \subseteq V \times V$ represents the set of undirected edges in the $k$-th state. In order to study the conditional correlation structure of $\bs Y_t$ within the $k$-th state through the graph $G_k$, we establish a useful result that allows us to make inference on the edge set $E_k$ based on the precision matrix $\bs \Theta_k$. 
Following \cite{finegold2011robust} and exploiting the mixture representation of the symmetric GH in \eqref{mixtureGH}, we can characterize the conditional correlation structure as shown in the following proposition.

\begin{proposition}\label{prop:cunc}
Consider $S_t = k$ and let $\bs {Y}_t \mid S_t = k \sim \mathcal{GH}_d(\bs \mu_k, \bs \Sigma_k, \lambda_k, \chi_k, \psi_k)$. If two nodes $j$ and $h$, with $j,h \in V$ and $j \neq h$, of the graph are separated by a set of nodes $C \in V$, then ${Y}_t^{(j)}$ and ${Y}_t^{(h)}$ are conditionally uncorrelated given $\bs {Y}_t^{(C)}$ and $S_t = k$. 
\end{proposition}

\begin{proof}
Without loss of generality, throughout the proof we omit the subscript $k$ indicating the latent state yet all the following equalities are conditional on $S_t = k$. To prove the result it is sufficient to show that ${Y}_t^{(j)}$ and ${Y}_t^{(h)}$ are conditionally uncorrelated given $\bs {Y}_t^{(V \backslash \{j,h \}) }$. Partition $V$ into $a=\{j,h \}$ and $b= V / \{j,h \}$. 
	For a given value of $W_t$ and given $\bs \Theta = \bs \Sigma^{-1}$:
	
\ber \label{prop1_proof}
\begin{aligned}
(\bs {Y}_t^{(a)} | \bs {Y}_t^{(b)},W_t) \sim  {\cal N}_2(& \bs{\mu}^{(a)} - \bs \Theta^{-1}_{a,a}\bs \Theta_{a,b}(\bs {Y}_t^{(b)} - \bs{\mu}^{(b)}), W_t \bs \Theta^{-1}_{a,a} )
\end{aligned}
\eer
	 and
\ber \label{prop1_proof2}
\begin{aligned}
({Y}_t^{(j)} | \bs {Y}_t^{(h\cup b)},W_t) \sim {\cal N}_1(&{\mu}^{(j)} - \Theta^{-1}_{j,j}\bs \Theta_{j,h \cup b}(\bs {Y}_t^{(h\cup b)} - \bs{\mu}^{(h\cup b)}), W_t \Theta^{-1}_{j,j}),
\end{aligned}
\eer
where $\bs \Theta_{a,b}$ is the submatrix of $\bs \Theta$ with rows and columns indexed by the sets $a$ and $b$. Since $\Theta_{j,h}=0$, 

\ber
\begin{aligned}
E[{Y}_t^{(j)} | \bs {Y}_t^{(h\cup b)},W_t] & = \mu^{(j)} - \Theta^{-1}_{j,j}\bs \Theta_{j,b} (\bs {Y}_t^{(b)} - \bs \mu^{(b)})\\ & = E[{Y}_t^{(j)} | \bs {Y}_t^{(b)},W_t]
\end{aligned}
\eer	 
for any value of $W_t$. Therefore,

\ber\label{eq:cond_mean}
\begin{aligned}
E[{Y}_t^{(j)} | \bs {Y}_t^{(h\cup b)}] & = E[ E[{Y}_t^{(j)} | \bs {Y}_t^{(h\cup b)},W_t] | \bs {Y}_t^{(h\cup b)}] = E[ E[{Y}_t^{(j)} | \bs {Y}_t^{(b)},W_t] | \bs {Y}_t^{(b)}]\\ & = E[{Y}_t^{(j)} | \bs {Y}_t^{(b)}]
\end{aligned}
\eer
which implies that ${Y}_t^{(j)}$ and ${Y}_t^{(h)}$ are conditionally uncorrelated given $\bs {Y}_t^{(b)}$.
\end{proof}

Based on \eqref{prop1_proof}, we can exploit the properties of GGMs to characterize conditional correlation relationships between the elements of of $\bs Y_t$ in each state $k=1,\dots,K$, estimating the precision matrices $ \bs \Theta_k$'s.
More formally, for each state and for any pair of nodes $j,h$, with $j \neq h$, $Y^{(j)}_t$ and $ Y^{(h)}_t$ are conditionally uncorrelated given $\bs {Y}_t^{(V \backslash \{j,h \}) }$ and $S_t = k$ if and only if the $(j,h)$-th element of the matrix $ \bs \Theta_k$, $\Theta_{k,j,h}$, is equal to 0. Hence, for each $k=1,\dots,K$, the edge set $E_k$ of the graph $G_k$ describing the distribution of $\bs Y_t \mid S_t = k$ is completely encoded by the matrix $\bs \Theta_k$ of the GH distribution, i.e., $(j,h) \in E_k$ if and only if $\Theta_{k,j,h} \neq 0$. The proposed methodology allows us to capture regime-specific conditional correlation structures for probability distributions within the GH family. 
 With respect to alternative strategies, our framework is a generalization of the graphical model of \cite{finegold2011robust} when $K = 1$ where the GH reduces to the multivariate t distribution. In addition, we encompass the graphical hidden Markov model of \cite{stadler2013penalized} when the state-specific GH emission densities reduce to the multivariate Normal distributions. 

 In this setting, it is therefore crucial to accurately determine the matrices $\bs \Theta_1, \dots, \bs \Theta_K$ for a correct interpretation of the graphs and to visualize the true interactions among the variables. In section \ref{sec:Est}, we introduce a maximum likelihood approach to estimate the model parameters and make inference on the sparsity pattern in the $\bs \Theta_k$'s.

\section{Estimation}\label{sec:Est}
 As shown in the previous sections, the location-scale mixture representation of the GH is a convenient tool to build a sparse graphical model. As is common for HMMs, and for latent variable models in general, we propose a suitable likelihood-based EM algorithm (\citealt{baum1970maximization}) to estimate the parameters of the method proposed based on the observed data. Since we are interested in detecting only the most important connections, we also propose a penalized version of the EM by considering the Lasso-type regularization of \cite{tib1996lasso} for sparse matrix estimation in high-dimensional settings where a large number of variables is available.  
 \subsection{The EM algorithm}\label{sub:EM}
 For a given number of hidden states $K$, the EM algorithm runs on the complete log-likelihood function of the model introduced, which is defined as	 
\be
\begin{aligned}\label{eq:completel}
\ell_c(\bs \theta) = \sum_{k=1}^K {\gamma_1}(k)\log\pi_k &+ \sum_{t=1}^T\sum_{k=1}^K \sum_{j=1}^K {\xi_t}(j,k) \log \pi_{k\lvert j} \\ &+ \sum_{t=1}^T\sum_{k=1}^K {\gamma_t}(k)\log f_Y(\bs y_t \lvert S_t = k),
\end{aligned}
\ee
where $\bs \theta = (\bs \mu_1, \dots, \bs \mu_K, \bs \Sigma_1, \dots, \bs \Sigma_K, \lambda_1, \dots, \lambda_K, \psi_1, \dots, \psi_K, \chi_1, \dots, \chi_K, \bs \pi, \bs \Pi)$ represents the vector of all model parameters, ${\gamma_t}(k)$ denotes a dummy variable equal to $1$ if the latent process is in state $k$ at occasion $t$ and 0 otherwise, and ${\xi_t}(j,k)$ is a dummy variable equal to $1$ if the process is in state $j$ in $t-1$ and in state $k$ at time $t$ and $0$ otherwise.\\
Unfortunately, the log-likelihood of the GH distribution does not yield any closed-form expression in the optimization process. Nevertheless, following \cite{chatzis2010hidden}, the issue can be resolved exploiting the data augmentation scheme of equation \eqref{GHrep2}. We exploit the conditional structure by writing
\be
\log f_Y(\bs y_t \lvert S_t = k) = l_1(\bs \mu_k, \bs \Sigma_k \lvert S_t = k) + l_2(\lambda_k, \chi_k, \psi_k \lvert S_t = k)
\ee
where 

\be\label{eq:l1}
l_1(\bs \mu_k, \bs \Sigma_k | S_t = k) = \sum_{t=1}^T \left[ -\frac{d}{2} \log(2\pi) - \frac{d}{2}\log(w_{tk}) - \frac{1}{2}\log \lvert \bs \Sigma_k \lvert - \frac{\delta(\bs y_t; \bs \mu_k, \bs \Sigma_k)}{2w_{tk}}  \right ]
\ee

and

{\footnotesize
\be\label{eq:l2}
l_2(\lambda_k, \chi_k, \psi_k \lvert S_t = k) = \sum_{t=1}^T \left \{ (\lambda_k - 1)\log(w_{tk}) - \frac{1}{2}\frac{\chi_k}{w_{tk}} - \frac{1}{2}\psi_k w_{tk} - \frac{1}{2}\lambda_k\log(\chi_k) + \frac{1}{2}\lambda_k\log(\psi_k) - \log \left [2K_{\lambda_k}\left (\sqrt{(\psi_k \chi_k)} \right) \right ]  \right \}
\ee
}

Working on $\ell_c(\bs \theta)$, we adopt the expectation-conditional maximization either (ECME) algorithm of \cite{liu1994ecme}. 
 The ECME algorithm is an extension of the expectation-conditional maximum (ECM) algorithm which 
 replaces the M-step of the EM algorithm by a number of computationally simpler conditional maximization (CM) steps. The ECME algorithm generalizes the ECM one by conditionally maximizing on some or all of the CM-steps the incomplete-data log-likelihood. In our case, the ECME algorithm iterates between four steps, one E-step and three CM-steps, until convergence. The three CM-steps arise from the update of the partition of $\bs \theta$ as $\{\bs \theta_1, \bs \theta_2, \bs \theta_3 \}$, where $\bs \theta_1 = \{\bs \mu_1, \bs \Sigma_1, \dots, \bs \mu_K, \bs \Sigma_K \}$, $\bs \theta_2 = \{\lambda_1, \chi_1, \psi_1, \dots, \lambda_K, \chi_K, \psi_K \}$ and $\bs \theta_3 = \{ \bs \pi, \bs \Pi \}$. 
 

\subsubsection*{E-step:} 
In the E-step, at the generic $(h+1)$-th iteration, the unobservable indicator variables ${\gamma_t}(k)$ and ${\xi_t}(j,k)$ in \eqref{eq:completel} are replaced by their conditional expectations given the observed data and the current parameter estimates $\bs \theta^{(h)}$. To compute such quantities we require the calculation of the probability of being in state $k$ at time $t$ given the observed sequence 
\be\label{gamma}
\gamma_t^{(h)}(k) = P_{\bs \theta^{(h)}}(S_t = k \lvert \bs y_1,\dots, \bs y_T )
\ee
and the probability that at time $t-1$ the process is in state $j$ and then in state $k$ at time $t$, given the observed sequence
\be\label{xi} 
\xi_t^{(h)}(j,k) = P_{\bs \theta^{(h)}}(S_{t-1} = j, S_t = k \lvert \bs y_1,\dots, \bs y_T ).
\ee
The quantities in \eqref{gamma} and \eqref{xi} can be obtained using the Forward-Backward algorithm of \cite{welch2003hidden}. 

Then, we use these to calculate the conditional expectation of \eqref{eq:completel}, given the observed data and the current estimates:
\be
\begin{aligned}\label{eq:estep}
Q(\bs \theta \lvert \bs \theta^{(h)}) = \sum_{k=1}^K \gamma^{(h)}_1(k)\log\pi_k &+ \sum_{t=1}^T\sum_{k=1}^K \sum_{j=1}^K \xi^{(h)}_t(j,k) \log \pi_{k\lvert j} \\ &+ \sum_{t=1}^T\sum_{k=1}^K \gamma^{(h)}_t(k) \left [ Q_1(\bs \mu_k, \bs \Sigma_k \lvert \bs \theta^{(h)}, S_t = k) + Q_2(\lambda_k, \chi_k, \psi_k \lvert \bs \theta^{(h)}, S_t = k) \right ].
\end{aligned}
\ee

In \eqref{eq:estep}, $Q_1(\bs \mu_k, \bs \Sigma_k \lvert \bs \theta^{(h)}, S_t = k)$ and $Q_2(\lambda_k, \chi_k, \psi_k \lvert \bs \theta^{(h)}, S_t = k)$ are respectively the conditional expectations of  $l_1(\bs \mu_k, \bs \Sigma_k \lvert S_t = k)$ and 
$l_2(\lambda_k, \chi_k, \psi_k \lvert S_t = k)$ given the observed data, using the current $\bs \theta^{(h)}$ for $\bs \theta$.  To compute $Q_1$ and $Q_2$, we have to consider the expected value of any function of the latent variable $W_{tk}$ in equation \eqref{eq:l1}. 
In particular the functions are $w_{tk}$, $1/w_{tk}$ and $\log(w_{tk})$ and their expected values must be calculated with respect to the current $\bs \theta^{(h)}$. Since the conditional distribution of $W_{tk}$ given $\bs Y_t$ corresponds to a GIG distribution with parameters $\lambda_k - \frac{d}{2}, \delta(\bs y_t; \bs \mu_k, \bs \Sigma_k) + \chi_k, \psi_k$, i.e.,

\be\label{eq:gig}
f_{W_{tk} |\bs Y_t}(w_{tk} | \bs Y_t= \bs y_t, S_t = k) \sim GIG(\lambda_k - \frac{d}{2}, \delta(\bs y_t; \bs \mu_k, \bs \Sigma_k) + \chi_k, \psi_k),
\ee

\noindent{it} follows that (see \citealt{chatzis2010hidden})

\be \label{eq:v_t}
{\small
v^{(h)}_{tk} = E^{(h)}
[W_{tk} \lvert \bs Y_t= \bs y_t, S_t = k] = \left(\frac{\delta(\bs y_t; \bs \mu_k^{(h)}, \bs \Sigma_k^{(h)}) + \chi_k^{(h)}}{\psi_k^{(h)}} \right)^{\frac{1}{2}} \frac{K_{\lambda_k^{(h)} - \frac{d}{2} +1} \left( \sqrt{\psi_k^{(h)} \left [ \delta(\bs y_t; \bs \mu_k^{(h)}, \bs \Sigma_k^{(h)}) + \chi_k^{(h)} \right ]} \right) }{K_{\lambda_k^{(h)} - \frac{d}{2}} \left( \sqrt{\psi_k^{(h)} \left [ \delta(\bs y_t; \bs \mu_k^{(h)}, \bs \Sigma_k^{(h)}) + \chi_k^{(h)} \right ]} \right) },
}
\ee

\be
{\small
\begin{aligned}\label{eq:u_t}
u^{(h)}_{tk} = E^{(h)}[W_{tk}^{-1}\lvert \bs Y_t= \bs y_t, S_t = k] &= \left(\frac{\delta(\bs y_t; \bs \mu_k^{(h)}, \bs \Sigma_k^{(h)}) + \chi_k^{(h)}}{\psi_k^{(h)}} \right)^{-\frac{1}{2}} \frac{K_{\lambda_k^{(h)} - \frac{d}{2} +1} \left( \sqrt{\psi_k^{(h)} \left [ \delta(\bs y_t; \bs \mu_k^{(h)}, \bs \Sigma_k^{(h)}) + \chi_k^{(h)} \right ]} \right) }{K_{\lambda_k^{(h)} - \frac{d}{2}} \left( \sqrt{\psi_k^{(h)} \left [ \delta(\bs y_t; \bs \mu_k^{(h)}, \bs \Sigma_k^{(h)}) + \chi_k^{(h)} \right ]} \right) } \\ &- \frac{2 \left ( \lambda_k^{(h)}  - \frac{d}{2} \right)}{\delta(\bs y_t; \bs \mu_k^{(h)}, \bs \Sigma_k^{(h)}) + \chi_k^{(h)}}
\end{aligned}
}
\ee

\noindent{and}
\be
\begin{aligned}\label{eq:z_t}
	z^{(h)}_{tk} = E^{(h)}[\log(W_{tk})\lvert \bs Y_t= \bs y_t, S_t = k] &= \log \left ( \sqrt{\frac{\delta(\bs y_t; \bs \mu_k^{(h)}, \bs \Sigma_k^{(h)}) + \chi_k^{(h)}}{\psi_k^{(h)}} }\right ) \\&+ \frac{\partial}{\partial{\lambda_k^{(h)}}}\log K_{\lambda_k^{(h)} - \frac{d}{2}} \left ( \sqrt{\psi_k^{(h)} \left [ \delta(\bs y_t; \bs \mu_k^{(h)}, \bs \Sigma_k^{(h)}) + \chi_k^{(h)} \right ]} \right ).
\end{aligned}
\ee

As a consequence, substituting $w_{tk}$, $1/w_{tk}$ and $\log(w_{tk})$ with $v^{(h)}_{tk}$, $u^{(h)}_{tk}$ and $z^{(h)}_{tk}$ respectively in $l_1(\bs \mu_k, \bs \Sigma_k | S_t = k)$ and $l_2(\lambda_k, \chi_k, \psi_k \lvert S_t = k)$ we obtain

\be\label{eq:Q1}
Q_1(\bs \mu_l, \bs \Sigma_k \lvert \bs \theta^{(h)}, S_t = k) = \sum_{t=1}^T \left[ - \frac{1}{2}\log \lvert \bs \Sigma_k \lvert - \frac{u^{(h)}_{tk} \delta(\bs y_t; \bs \mu_k, \bs \Sigma_k)}{2}  \right ]
\ee

{\footnotesize
\be
\begin{aligned}\label{eq:Q2}
Q_2(\lambda_k, \chi_k, \psi_k \lvert \bs \theta^{(h)}, S_t = k) = \sum_{t=1}^T \left \{ (\lambda_k - 1)z^{(h)}_{tk} - u^{(h)}_{tk}\frac{\chi_k}{2} - \frac{\psi_k}{2} v^{(h)}_{tk} - \frac{\lambda_k}{2} \left [\log(\chi_k) + \log(\psi_k) \right ] - \log \left [2K_{\lambda_k}\left (\sqrt{(\psi_k \chi_k)} \right) \right ]  \right \}
\end{aligned}
\ee
}

where in \eqref{eq:Q1} we dropped the terms which are constant with respect to $\bs \mu_k$ and $\bs \Sigma_k$.

\subsubsection*{CM-step 1:} 

The initial probabilities $\pi_k$ and transition probabilities $\pi_{k\lvert j}$ of the partition $\bs \theta_3$ are updated using:
\be
\pi^{(h+1)}_k = \gamma^{(h)}_1(k), \quad k = 1,\dots,K
\ee
and
\be
\pi^{(h+1)}_{k\lvert j} = \frac{\sum_{t=1}^T \xi^{(h)}_t(j,k)}{\sum_{t=1}^T \sum_{k=1}^K \xi^{(h)}_t(j,k)}, \quad j,k = 1,\dots,K.
\ee

\subsubsection*{CM-step 2:} 
The second maximization step requires the calculation of $\bs \theta^{(h + 1)}_1$ as the value of $\bs \theta_1$ that maximizes $Q_1(\bs \mu_k, \bs \Sigma_k \lvert \bs \theta^{(h)}, S_t = k)$ in \eqref{eq:Q1}, with $\bs \theta_2$ fixed at $\bs \theta^{(h + 1)}_2$. The first order conditions with respect to $\bs \mu_k$ and $\bs \Theta_k = \bs \Sigma_k^{-1}$ yield

\be\label{eq:1stordermu}
\frac{\partial{Q_1(\bs \mu_k, \bs \Sigma_k \lvert \bs \theta^{(h)}_2)}}{\partial{\bs \mu_k}} = \sum_{t=1}^T \sum_{k=1}^K u^{(h)}_{tk} (\bs y_t - \bs \mu_k) \gamma^{(h)}_t(k)
\ee

\be\label{eq:1storderSigma}
\frac{\partial{Q_1(\bs \mu_k, \bs \Theta_k \lvert \bs \theta^{(h)}_2)}}{\partial{\bs \Theta_k}} = \sum_{t=1}^T \sum_{k=1}^K \gamma^{(h)}_t(k) \left [ \frac{1}{2}\frac{\partial}{\partial{\bs \Theta_k}}\log \lvert \bs \Theta_k \lvert	 - \frac{1}{2}u^{(h)}_{tk}\frac{\partial}{\partial{\bs \Theta_k}} (\bs y_t - \bs \mu_k)'\bs \Theta_k(\bs y_t - \bs \mu_k) \right ]
\ee

so the CM-step update expressions for $\bs \mu_k$ and $\bs \Sigma_k$ are 

\be\label{eq:mu}
\bs \mu_k^{(h+1)} = \frac{\sum_{t=1}^T \sum_{k=1}^K \gamma^{(h)}_t(k) u^{(h)}_{tk} \bs y_t}{\sum_{t=1}^T \sum_{k=1}^K \gamma^{(h)}_t(k) u^{(h)}_{tk}}
\ee

\noindent{and} 

\be\label{eq:Sigma}
\bs \Sigma_k^{(h+1)} = \lvert \bs \Sigma_k^{*(h+1)}\lvert^{-\frac{1}{d}} \bs \Sigma_k^{*(h+1)}
\ee

\noindent{where}

\be\label{eq:Sigma*}
\bs \Sigma_k^{*(h+1)} = \frac{\sum_{t=1}^T \sum_{k=1}^K \gamma^{(h)}_t(k) u^{(h)}_{tk} (\bs y_t - \bs \mu_k^{(h+1)})'(\bs y_t - \bs \mu_k^{(h+1)})}{\sum_{t=1}^T \sum_{k=1}^K \gamma^{(h)}_t(k)}
\ee

In \eqref{eq:Sigma}, the scalar $\lvert \bs \Sigma_k^{*(h+1)}\lvert^{-\frac{1}{d}}$ is needed to ensure the identifiability constraint  $\lvert \bs \Sigma_k^{(h+1)} \lvert = 1$.

\subsubsection*{Cm-step 3:} 
In the third CM-step we choose the value of $\bs \theta_2$ that maximizes $\ell_c(\bs \theta)$ in \eqref{eq:completel}, with $\bs \theta_1$ fixed at $\bs \theta^{(h+1)}_1$. As a closed-form solution for $\bs \theta^{(h+1)}_2$ is not analytically available, numerical optimization 
 can be used with this aim. As in \cite{bagnato2023generalized}, we perform an unconstrained maximization on $\mathbb{R}^{3}$, based on a (log/exp) transformation/back-transformation approach for $\chi_k$ and $\psi_k$.\\

The E- and CM-steps are alternated until convergence, that is when the observed likelihood between two consecutive iterations is smaller than a predetermined threshold. In this paper, we set this threshold criterion equal to $10^{-8}$. 

For fixed $K$ we initialize the ECME algorithm by providing the initial states partition, $\{ S_t^{(0)} \}_{t=1}^T$, according to the $K$-means algorithm. From the generated partition, the elements of $\bs \Pi^{(0)}$ are computed as proportions of transition. 
The location parameters $\bs \mu_1, \dots, \bs \mu_K$ are obtained from the centroid of the the derived clusters while we set the initial values of the $\bs \Sigma_k$'s as the empirical covariance matrices of the obtained clusters with unit determinant using \eqref{eq:Sigma}. Finally, the $(\lambda_k, \ \chi_k, \ \psi_k)$'s are initialized from uniform distributions.
 To deal with the possibility of multiple roots of the likelihood equation and better explore the parameter space, 
 we fit the proposed HMGHGM using a multiple random starts strategy with different starting partitions and retain the solution corresponding to the maximum likelihood value.
 
 All the computations have been conducted using the R software, version 4.3.0 (\citealt{R}).


\subsection{Penalized inference with the GH distribution}\label{sec:pen}
 In this section we extend the procedure described above for estimating high-dimensional graphs where the number of model parameters grows with the dimension of the problem. Indeed, the correct identification of the sparsity patterns is a fundamental issue for capturing the most relevant interconnections, which motivates us to use sparse estimators that automatically shrink the elements of the inverse of the  scale matrix matrix. In the literature, within the GGMs context, several works have been put forward to obtain sparse estimates of the concentration matrix. \cite{friedman2008sparse} introduced the efficient coordinate descent algorithm, denoted as \textit{glasso}, to estimate a sparse graph using the Lasso $L_1$ penalty of \cite{tib1996lasso}. \cite{gao2015estimation} proposed to estimate the concentration matrix using further penalties, i.e., the smoothly clipped absolute deviation of \cite{fan2001variable} and the minimax concave penalty of \cite{zhang2010nearly}. Alternatively \cite{meinshausen2006high} presented a neighbourhood selection scheme that estimates the conditional independence relations separately for each node in the graph by using $L_1$-penalized regressions.\\
In this work, we exploit the efficient \textit{glasso} approach of \cite{friedman2008sparse} to induce sparsity in the precision matrices of the GH for each latent state, without relying on the limitation of normally distributed returns. Specifically, starting from the EM algorithm of Section \ref{sec:Est} and following \cite{green1990use}, we construct a PEM algorithm by adding to the complete likelihood in \eqref{eq:completel} an $L_1$-norm penalty that shrinks to zero the off-diagonal elements of the $\bs \Theta_k^{-1}$'s.
The penalized version of the conditional Normal log-likelihood in \eqref{eq:l1} is reported in the following definition.

\begin{definition}\label{prop:completeloglikpen}
The conditional Normal log-likelihood function is proportional to:
\begin{equation}\label{eq:completelpen}
\ell_{1,pen}(\bs \mu_k, \bs \Theta_k| \bs Y_t, W_{tk}, S_t = k) \propto \frac{1}{2} \sum_{t=1}^T \gamma_t(k) \log | \bs \Theta_k | - \frac{1}{2} \textnormal{tr} \{ \tilde{\bs S}_k \bs \Theta_k  \} -\rho \sqrt{\nu_k} ||\bs \Theta_k ||_1
\end{equation}
where $|| \bs \Theta_k ||_1$ is the sum of the absolute values of the off-diagonal entries of the matrix $\bs \Theta_k$ and $\tilde{\bs S}_k$ is the weighted empirical covariance matrix defined as
\begin{equation}\label{eq:Smat}
\tilde{\bs S}_k = \sum_{t=1}^T \frac{\gamma_t(k)}{W_{tk}} (\bs Y_t - \boldsymbol{\mu}_k) (\bs Y_t - \boldsymbol{\mu}_k)'.
\end{equation}

Additionally, $\nu_k > 0$ is a state-dependent weight to allow for a different penalty in each latent state such that $\sum_{k=1}^K \nu_k = 1$ and $\rho \geq 0$ is the tuning parameter.
\end{definition}

As one can see, the penalized complete likelihood in \eqref{eq:completelpen} is exactly the objective function maximized in a GGM where the empirical covariance matrix of the data is substituted by the matrix $\tilde{\bs S}_k$ defined in \eqref{eq:Smat}, for $k = 1,\dots,K$. This suggests that we can exploit the simple and fast \textit{glasso} algorithm of \cite{friedman2008sparse} to update the estimate of $\bs \Theta_k$. More specifically, for a given value of $\rho$, the proposed PEM alternates between the E-step described in Section \ref{sub:EM} and essentially modifies the CM-step 2 by maximizing the quantity in \eqref{eq:completelpen}.

 
\section{Simulation study}\label{sec:sim}
We conduct a simulation study to validate the performance of our model under different scenarios in terms of:
\begin{enumerate}
\item[1.] recovering the true values of the parameters for a varying number of hidden states, 
\item[2.] clustering performance,
\item[3.] the ability to correctly retrieve the edge sets associated to each hidden state. 
\end{enumerate}
 
\noindent{\textit{Scenario 1:}} In the first scenario we consider a $2$-dimensional vector $\bs Y_t$ with sample size of $T = 1000$, $150$ Monte Carlo simulations, $70$ random starts for the EM algorithm and we draw observations from a one-, two- and three-state HMM. 
 Conditional on the hidden state, observations are generate from the following six models: Gaussian, t, Cauchy, Laplace, generalized hyperbolic and variance gamma (see Table \ref{tab:DGM}). They all share the same location and scale parameters $\bs \mu_1 = [5,5],\ \bs \Sigma_1 = [1.51, -1.13, -1.13, 1.51]$, $\bs \mu_2 = [-5,-5],\ \bs \Sigma_2 = [1.51, 1.13, 1.13, 1.51]$, $\bs \mu_3 = [0,0],\ \bs \Sigma_3 = [1.01, .12, .12, 1.01]$. Parameters $\lambda,\ \chi, \ \psi$ are state-invariant and they are chosen so to obtain the six models in Table \ref{tab:DGM}. In this first scenario, to fit the proposed model we use the non penalized ECME algorithm discussed in Section \ref{sub:EM}. The matrices of transition probabilities and the vectors of initial probabilities are set equal to $\bs \Pi = \bigl( \begin{smallmatrix} 0.9 & 0.1\\ 0.1 & 0.9\end{smallmatrix}\bigr)$ and $\bs \pi = [0.7, 0.3]$ for $K=2$, and to $\bs \Pi = \bigl( \begin{smallmatrix} 0.8 & 0.1 & 0.1\\ 0.1 & 0.8 & 0.1\\ 0.1 & 0.1 & 0.8\\ \end{smallmatrix}\bigr)$ and $\bs \pi = [ 0.4, 0.3, 0.3]$ for $K=3$, respectively. 
 In order to assess the validity of our model we compute the point estimate and standard error values associated to the state-specific coefficients, averaged over the Monte Carlo replications. 
 Tables \ref{tab:normsim}, \ref{tab:tsim}, \ref{tab:lapsim}, \ref{tab:cauchysim}, \ref{tab:ghsim}, \ref{tab:vgsim} illustrate the results for the six data generating models considered. For every model, the estimated location and scale parameters are very close to the true values. This is also true for the parameters ($\lambda_k, \chi_k, \psi_k$) in every state. In particular, for the Gaussian generating model where $\lambda_k$ 
 is constrained so that $\lambda_k = -\chi_k/2$, we observe large values for $\chi_k$ whereas $\psi_k$ is approximately zero. For all the other distributions, the point estimates of $\lambda_k$ are always very close to the real values in Table \ref{tab:DGM}. The estimated $\psi_k$'s are small for the t and the Cauchy, while, as expected, are very close to 0.5 for the Cauchy and variance gamma. Finally, as regards the $\chi_k$'s, they are close to 2 for the t, the Cauchy and the Generalized Hyperbolic, and near to zero for the Laplace and variance gamma. 


\begin{table}[htbp]
  \centering
    \begin{tabular}{lrrrrrr}
    Parameters & Gaussian & t     & Cauchy & Laplace & generalized hyperbolic & variance gamma \\
    \midrule
    $\lambda$ & -20   & -1    & -0.5  & 1     & (d+1)/2    & 1.5 \\
    $\chi$ & 40    & 2     & 2     & 0.001 & 2     & 0.001 \\
    $\psi$ & 0.001 & 0.001 & 0.001 & 0.5   & 3     & 0.5 \\
    \bottomrule
    \end{tabular}%
    \caption{Values of the parameters $\lambda,\ \chi, \ \psi$ used in the simulation study.}
  \label{tab:DGM}%
\end{table}%
  
 \begin{table}[htbp]
  \centering
\begin{tabular}{lrrrrrrrrr}
      & $\bs \mu_k$ &       & $\bs \Sigma_k$ &       &       &       & $\lambda$ & $\chi$ & $\psi$ \\
\midrule
\textit{$K=1$} &       &       &       &       &       &       &       &       &  \\
Est.  & 4.982 & 5.008 & 1.506 & -1.123 & -1.123 & 1.507 & -8.261 & 16.522 & 0.139 \\
Std. Err. & 0.061 & 0.054 & 0.05  & 0.053 & 0.053 & 0.051 & 1.514 & 3.027 & 0.059 \\
\textit{$K=2$} &       &       &       &       &       &       &       &       &  \\
Est.  & -4.99 & -4.984 & 1.5   & -1.119 & -1.119 & 1.502 & -7.472 & 14.945 & 0.038 \\
Std. Err. & 0.086 & 0.079 & 0.072 & 0.075 & 0.075 & 0.078 & 6.828 & 13.656 & 0.134 \\
Est.  & 4.987 & 5.007 & 1.527 & 1.147 & 1.147 & 1.517 & -11.16 & 22.321 & 0.002 \\
Std. Err. & 0.086 & 0.09  & 0.074 & 0.071 & 0.071 & 0.072 & 5.827 & 11.655 & 0.135 \\
\textit{$K=3$} &       &       &       &       &       &       &       &       &  \\
Est.  & -5.102 & -5.098 & 1.51  & -1.137 & -1.137 & 1.518 & -7.961 & 15.922 & 0.003 \\
Std. Err. & 0.142 & 0.133 & 0.086 & 0.095 & 0.095 & 0.096 & 6.255 & 12.509 & 0.091 \\
Est.  & -0.02 & -0.001 & 1.018 & 0.112 & 0.112 & 0.995 & -21.538 & 43.076 & 0.067 \\
Std. Err. & 0.107 & 0.102 & 0.062 & 0.07  & 0.07  & 0.06  & 8.968 & 17.937 & 0.053 \\
Est.  & 4.969 & 5.005 & 1.441 & 1.04  & 1.04  & 1.445 & -7.719 & 15.438 & 0.028 \\
Std. Err. & 0.114 & 0.114 & 0.094 & 0.09  & 0.09  & 0.093 & 8.088 & 11.175 & 0.069 \\
\bottomrule
\end{tabular}%
    \caption{Point estimate (Est.) and standard error (Std. Err.) values of the parameters of the Gaussian data generating model, with $T=1000$ observations for all of the three states of the HMM.}
  \label{tab:normsim}%
\end{table}%

  \begin{table}[htbp]
  \centering
\begin{tabular}{lrrrrrrrrr}
      & $\bs \mu_k$ &       & $\bs \Sigma_k$ &       &       &       & $\lambda$ & $\chi$ & $\psi$ \\
\midrule
\textit{$K=1$} &       &       &       &       &       &       &       &       &  \\
Est.  & 4.995 & 4.998 & 1.506 & -1.125 & -1.125 & 1.51  & -0.973 & 1.936 & 0.002 \\
Std. Err. & 0.042 & 0.042 & 0.066 & 0.057 & 0.057 & 0.058 & 0.107 & 0.256 & 0.013 \\
\textit{$K=2$} &       &       &       &       &       &       &       &       &  \\
Est.  & -5.003 & -4.996 & 1.5   & -1.119 & -1.119 & 1.502 & -0.952 & 1.893 & 0.002 \\
Std. Err. & 0.065 & 0.067 & 0.09  & 0.081 & 0.081 & 0.083 & 0.149 & 0.379 & 0.02 \\
Est.  & 4.99  & 5.006 & 1.529 & 1.159 & 1.159 & 1.532 & -0.952 & 1.857 & 0.002 \\
Std. Err. & 0.061 & 0.063 & 0.094 & 0.087 & 0.087 & 0.085 & 0.149 & 0.378 & 0.022 \\
\textit{$K=3$} &       &       &       &       &       &       &       &       &  \\
Est.  & -5.011 & -4.993 & 1.491 & -1.107 & -1.107 & 1.493 & -0.944 & 1.821 & 0.002 \\
Std. Err. & 0.077 & 0.081 & 0.111 & 0.108 & 0.108 & 0.096 & 0.299 & 0.56  & 0.082 \\
Est.  & 0.002 & 0.004 & 1.006 & 0.123 & 0.123 & 1.009 & -0.912 & 1.758 & 0.003 \\
Std. Err. & 0.067 & 0.068 & 0.078 & 0.077 & 0.077 & 0.081 & 0.222 & 0.504 & 0.045 \\
Est.  & 4.982 & 5.02  & 1.484 & 1.103 & 1.103 & 1.493 & -0.941 & 1.873 & 0.003 \\
Std. Err. & 0.076 & 0.081 & 0.105 & 0.109 & 0.109 & 0.12  & 0.298 & 0.569 & 0.081 \\
\bottomrule
\end{tabular}%
    \caption{Point estimate (Est.) and standard error (Std. Err.) values of the parameters of the t data generating model, with $T=1000$ observations for all of the three states of the HMM.}
  \label{tab:tsim}%
\end{table}%

 \begin{table}[htbp]
  \centering
\begin{tabular}{lrrrrrrrrr}
      & $\bs \mu_k$ &       & $\bs \Sigma_k$ &       &       &       & $\lambda$ & $\chi$ & $\psi$ \\
\midrule
\textit{$K=1$} &       &       &       &       &       &       &       &       &  \\
Est.  & 5.002 & 4.992 & 1.508 & -1.127 & -1.127 & 1.507 & -0.501 & 1.978 & 0.001 \\
Std. Err. & 0.064 & 0.064 & 0.062 & 0.066 & 0.066 & 0.061 & 0.047 & 0.288 & 0.001 \\
\textit{$K=2$} &       &       &       &       &       &       &       &       &  \\
Est.  & -4.983 & -4.983 & 1.513 & -1.125 & -1.125 & 1.497 & -0.484 & 1.933 & 0.001 \\
Std. Err. & 0.102 & 0.096 & 0.084 & 0.083 & 0.083 & 0.089 & 0.076 & 0.424 & 0.001 \\
Est.  & 4.991 & 5.011 & 1.496 & 1.121 & 1.121 & 1.509 & -0.483 & 1.881 & 0.001 \\
Std. Err. & 0.094 & 0.099 & 0.09  & 0.096 & 0.096 & 0.101 & 0.084 & 0.439 & 0.002 \\
\textit{$K=3$} &       &       &       &       &       &       &       &       &  \\
Est.  & -4.988 & -4.974 & 1.481 & -1.104 & -1.104 & 1.499 & -0.481 & 1.901 & 0.001 \\
Std. Err. & 2.648 & 3.838 & 0.16  & 0.161 & 0.161 & 0.143 & 1.233 & 3.551 & 0.022 \\
Est.  & -0.012 & 0     & 1.018 & 0.14  & 0.14  & 1.001 & -0.461 & 1.8   & 0.001 \\
Std. Err. & 0.327 & 0.521 & 0.136 & 0.231 & 0.231 & 0.164 & 1.638 & 3.642 & 0.015 \\
Est.  & 4.986 & 5.01  & 1.497 & 1.131 & 1.131 & 1.523 & -0.472 & 1.834 & 0.001 \\
Std. Err. & 0.119 & 1.911 & 0.131 & 0.144 & 0.144 & 0.154 & 0.212 & 1.479 & 0.009 \\
\bottomrule
\end{tabular}%
    \caption{Point estimate (Est.) and standard error (Std. Err.) values of the parameters of the Cauchy data generating model, with $T=1000$ observations for all of the three states of the HMM.}
  \label{tab:cauchysim}%
\end{table}%

 \begin{table}[htbp]
  \centering
\begin{tabular}{lrrrrrrrrr}
      & $\bs \mu_k$ &       & $\bs \Sigma_k$ &       &       &       & $\lambda$ & $\chi$ & $\psi$ \\
\midrule
\textit{$K=1$} &       &       &       &       &       &       &       &       &  \\
Est.  & 4.993 & 5.008 & 1.509 & -1.13 & -1.13 & 1.505 & 0.896 & 0     & 0.446 \\
Std. Err. & 0.049 & 0.049 & 0.057 & 0.053 & 0.053 & 0.052 & 0.12  & 0.105 & 0.062 \\
\textit{$K=2$} &       &       &       &       &       &       &       &       &  \\
Est.  & -4.986 & -5.002 & 1.504 & -1.116 & -1.116 & 1.493 & 0.824 & 0     & 0.41 \\
Std. Err. & 0.072 & 0.074 & 0.08  & 0.085 & 0.085 & 0.089 & 0.165 & 0.118 & 0.085 \\
Est.  & 4.986 & 5.001 & 1.496 & 1.119 & 1.119 & 1.505 & 0.81  & 0.004 & 0.401 \\
Std. Err. & 0.075 & 0.074 & 0.085 & 0.082 & 0.082 & 0.08  & 0.227 & 0.26  & 0.092 \\
\textit{$K=3$} &       &       &       &       &       &       &       &       &  \\
Est.  & -4.988 & -4.999 & 1.511 & -1.135 & -1.135 & 1.514 & 0.754 & 0.001 & 0.36 \\
Std. Err. & 0.089 & 0.092 & 0.108 & 0.109 & 0.109 & 0.099 & 0.245 & 0.276 & 0.119 \\
Est.  & -0.007 & -0.003 & 1.007 & 0.137 & 0.137 & 1.012 & 0.756 & 0     & 0.368 \\
Std. Err. & 0.07  & 0.076 & 0.068 & 0.076 & 0.076 & 0.07  & 0.188 & 0.21  & 0.095 \\
Est.  & 4.984 & 5.012 & 1.476 & 1.099 & 1.099 & 1.496 & 0.741 & 0     & 0.352 \\
Std. Err. & 0.092 & 0.091 & 0.11  & 0.105 & 0.105 & 0.099 & 0.271 & 0.35  & 0.127 \\
\bottomrule
\end{tabular}%
    \caption{Point estimate (Est.) and standard error (Std. Err.) values of the parameters of the Laplace data generating model, with $T=1000$ observations for all of the three states of the HMM.}
  \label{tab:lapsim}%
\end{table}%

 \begin{table}[htbp]
  \centering
\begin{tabular}{lrrrrrrrrr}
      & $\bs \mu_k$ &       & $\bs \Sigma_k$ &       &       &       & $\lambda$ & $\chi$ & $\psi$ \\
\midrule
\textit{$K=1$} &       &       &       &       &       &       &       &       &  \\
Est.  & 4.991 & 5.005 & 1.51  & -1.12 & -1.12 & 1.501 & 1.48  & 2.068 & 2.849 \\
Std. Err. & 0.044 & 0.042 & 0.05  & 0.054 & 0.054 & 0.055 & 2.782 & 3.666 & 1.767 \\
\textit{$K=2$} &       &       &       &       &       &       &       &       &  \\
Est.  & -5.003 & -4.996 & 1.042 & -0.3  & -0.3  & 1.046 & 1.889 & 1.743 & 2.987 \\
Std. Err. & 0.05  & 0.054 & 0.052 & 0.059 & 0.059 & 0.057 & 3.46  & 4.566 & 2.461 \\
Est.  & 4.991 & 4.999 & 1.051 & 0.318 & 0.318 & 1.048 & 1.631 & 1.563 & 2.868 \\
Std. Err. & 0.054 & 0.054 & 0.05  & 0.049 & 0.049 & 0.049 & 3.136 & 4.315 & 1.979 \\
\textit{$K=3$} &       &       &       &       &       &       &       &       &  \\
Est.  & -5    & -5    & 1.486 & -1.102 & -1.102 & 1.491 & 1.883 & 0.968 & 2.644 \\
Std. Err. & 0.079 & 0.071 & 0.088 & 0.094 & 0.094 & 0.101 & 3.472 & 4.164 & 2.69 \\
Est.  & -0.01 & 0.014 & 1     & 0.112 & 0.112 & 1.012 & 2.078 & 1.018 & 2.854 \\
Std. Err. & 0.068 & 0.065 & 0.059 & 0.068 & 0.068 & 0.059 & 3.49  & 4.286 & 2.734 \\
Est.  & 4.995 & 5.015 & 1.524 & 1.144 & 1.144 & 1.514 & 1.793 & 2.045 & 2.603 \\
Std. Err. & 0.089 & 0.088 & 0.106 & 0.099 & 0.099 & 0.103 & 3.914 & 4.688 & 3.203 \\
\bottomrule
\end{tabular}%
    \caption{Point estimate (Est.) and standard error (Std. Err.) values of the parameters of the generalized hyperbolic data generating model, with $T=1000$ observations for all of the three states of the HMM.}
  \label{tab:ghsim}%
\end{table}%

 \begin{table}[htbp]
  \centering
\begin{tabular}{lrrrrrrrrr}
      & $\bs \mu_k$ &       & $\bs \Sigma_k$ &       &       &       & $\lambda$ & $\chi$ & $\psi$ \\
\midrule
\textit{$K=1$} &       &       &       &       &       &       &       &       &  \\
Est.  & 4.993 & 4.992 & 1.509 & -1.127 & -1.127 & 1.507 & 1.393 & 0.034 & 0.471 \\
Std. Err. & 0.068 & 0.068 & 0.053 & 0.056 & 0.056 & 0.053 & 0.409 & 1.105 & 0.093 \\
\textit{$K=2$} &       &       &       &       &       &       &       &       &  \\
Est.  & -5.004 & -4.997 & 1.505 & -1.137 & -1.137 & 1.524 & 1.395 & 0.013 & 0.467 \\
Std. Err. & 0.101 & 0.101 & 0.081 & 0.081 & 0.081 & 0.082 & 0.625 & 1.851 & 0.138 \\
Est.  & 4.994 & 5.012 & 1.515 & 1.136 & 1.136 & 1.511 & 1.269 & 0.077 & 0.428 \\
Std. Err. & 0.109 & 0.105 & 0.07  & 0.073 & 0.073 & 0.072 & 0.654 & 2.142 & 0.134 \\
\textit{$K=3$} &       &       &       &       &       &       &       &       &  \\
Est.  & -5.02 & -4.996 & 1.501 & -1.114 & -1.114 & 1.492 & 1.297 & 0.024 & 0.444 \\
Std. Err. & 0.139 & 0.145 & 0.09  & 0.101 & 0.101 & 0.107 & 0.846 & 2.935 & 0.184 \\
Est.  & -0.005 & 0.009 & 1.019 & 0.128 & 0.128 & 0.998 & 1.231 & 0.107 & 0.419 \\
Std. Err. & 0.119 & 0.115 & 0.074 & 0.083 & 0.083 & 0.073 & 0.751 & 2.32  & 0.204 \\
Est.  & 4.996 & 5.011 & 1.528 & 1.146 & 1.146 & 1.514 & 1.294 & 0.022 & 0.429 \\
Std. Err. & 0.132 & 0.138 & 0.11  & 0.095 & 0.095 & 0.1   & 0.893 & 3.113 & 0.188 \\
\bottomrule
\end{tabular}%
    \caption{Point estimate (Est.) and standard error (Std. Err.) values of the parameters of the variance gamma data generating model, with $T=1000$ observations for all of the three states of the HMM.}
  \label{tab:vgsim}%
\end{table}%


 \noindent{\textit{Scenario 2:}} To evaluate the ability in recovering the true states partition we consider the Adjusted Rand Index (ARI) of \cite{hubert1985comparing}. The state partition provided by the fitted models is obtained by taking the maximum, $\underset{k}{\max} \, \gamma_{t} (k)$, a posteriori probability for every $t = 1,\dots,T$. In Table \ref{tab:ARI} we report the mean and standard deviation of ARI for the posterior probabilities for the two ($K = 2$) and three-state ($K = 3$) HMGHGMs across the six settings considered.
  Firstly, we observe that the data generating model plays a fundamental role in estimating the true states partition, as we obtain slightly worst clustering performance for the model with the Cauchy distribution. Secondly, the goodness of the clustering obtained depends, as expected, on the number of states of the Markov chain considered, being the values slightly higher with $K = 2$ than with $K = 3$. Overall, the proposed HMGHGM is able to recover the true state partition highly satisfactory in all the cases examined.



\begin{table}[htbp]
  \centering
\begin{tabular}{lrrrrrr}
      &       &       & $K = 2$ &       &       &  \\
      & N     & t     & C     & L     & gh    & vg \\
\midrule
Mean  & 0.9989 & 0.9851 & 0.943 & 0.9909 & 0.9999 & 0.9833 \\
Std. Err. & 0.0022 & 0.0076 & 0.0157 & 0.0064 & 0.0005 & 0.0081 \\
      &       &       & $K = 3$ &       &       &  \\
Mean  & 0.9322 & 0.9228 & 0.7709 & 0.9059 & 0.9786 & 0.8506 \\
Std. Err. & 0.0171 & 0.0165 & 0.0746 & 0.0186 & 0.0087 & 0.0249 \\
\bottomrule
\end{tabular}%
    \caption{Mean and standard error (Std. Err.) values of ARI for the posterior probabilities for Gaussian (N), t, Cauchy (C), Laplace (L), generalized hyperbolic (gh) and variance gamma (vg) with $T=1000$ for $K = 2$ and $K = 3$.}
  \label{tab:ARI}%
\end{table}%

 \FloatBarrier
 
\noindent{\textit{Scenario 3:}}  Finally, in order to evaluate the ability of the HMGHGM to correctly retrieve the graph structure in each latent state, following \cite{finegold2011robust}, we draw a state-specific random $10 \times 10$ sparse matrix $\bs \Theta_k = \bs \Sigma_k^{-1}$ according to the following procedure:
\begin{itemize}
	\item Choose each lower-triangular element of $\bs \Theta_k$ independently to be $-1,0,1$ with probability $15\%, \ 70\%$ and $15\%$, respectively;
	\item For $j > h$ set $\Theta_{k,h,j} = \Theta_{k,j,h}$;
	\item To ensure positive definiteness of $\bs \Theta_k$, set $\Theta_{k,i,i} = 1 +h$, with $h$ being the number of nonzero elements in the $i$th row of $\bs \Theta_k$.
\end{itemize}
As in \cite{finegold2011robust}, to strengthen the partial correlations in $\bs \Theta_k$ we reduce its diagonal elements by fixing a minimum eigenvalue of $0.6$. We then sample $T = 1000$ observations from the data generating models described above (see Table \ref{tab:DGM}) for a one ($K = 1$) and two ($K = 2$) state HMM. For each of the six distributions, we run the proposed HMGHGM on an equispaced grid of $50$ values for the tuning parameter $\rho$ from $0.01$ to $0.9$ with the weight $\nu_k = 1/K$, $k=1,\dots,K$, assigning the same penalty for all states. For each model, we carry out $50$ Monte Carlo replications and assess how well different distributions of the GH family recover the true edges in each state by reporting ROC curves averaged over the $50$ replicates and latent states.
 From left to right, Figure \ref{fig:ROCs} shows the ROC curves, where we plot the true positives rate against the false positives rate for each value of $\rho$, under the multivariate Gaussian, t, Cauchy, Laplace, generalized hyperbolic and variance gamma multivariate distributions. Overall, these results suggest that when we deal with non-Gaussian heavy tailed data, our static and dynamic graphical models have the same performance of the graphical lasso, representing a robust alternative to existing static approaches  and a novelty for dynamic graphical models in non-Gaussian settings.

 \begin{figure}[h!]
\begin{center}
  \includegraphics[scale=.6]{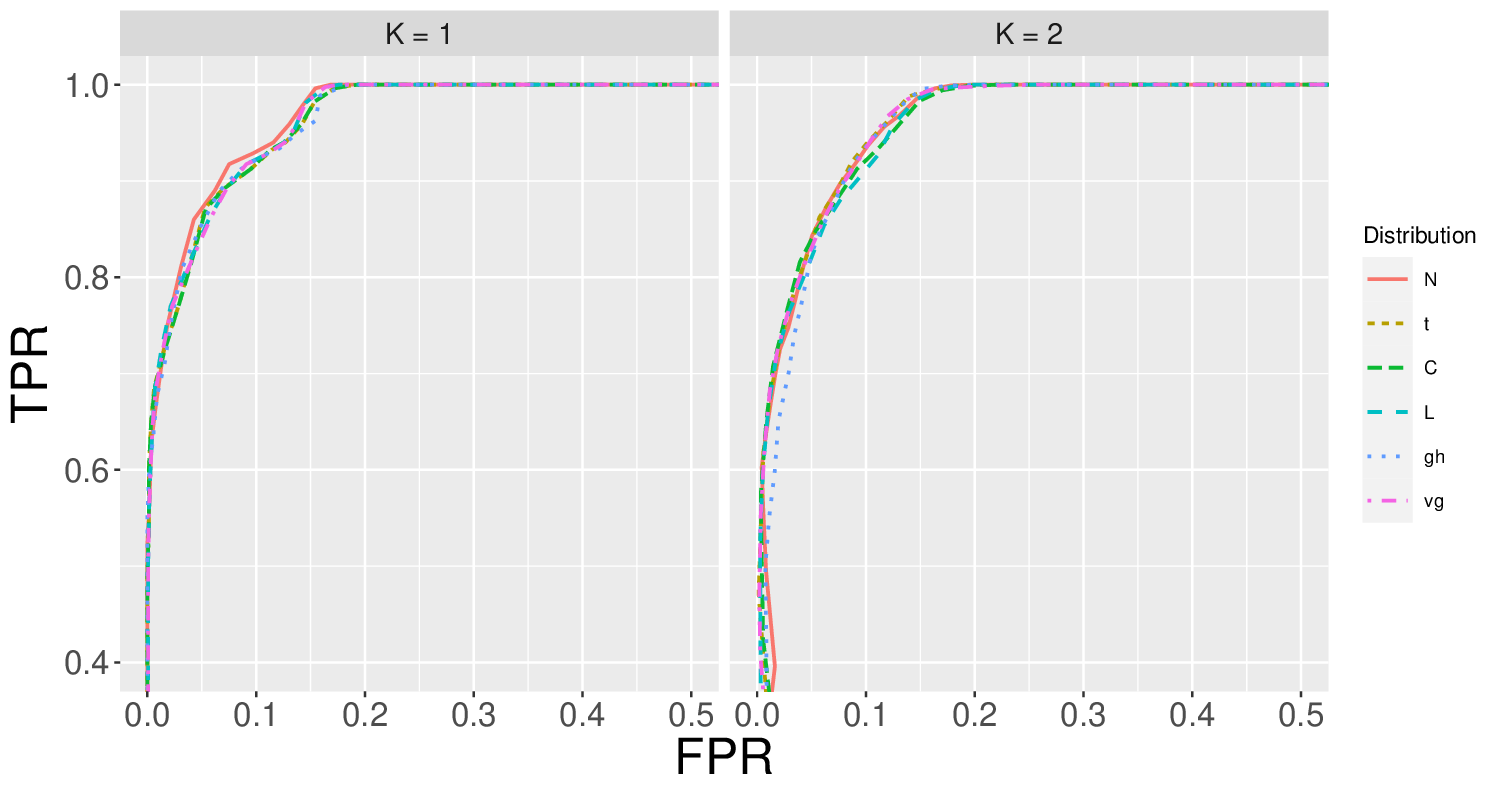}
\caption{ROC curves for the HMGHGM model under the multivariate Gaussian (N), t, Cauchy (C), Laplace (L), generalized hyperbolic (gh) and variance gamma (vg) distributions for $K = 1$ (left) and $K = 2$ (right) states. Each curve is averaged over 50 Monte Carlo replications. For the two-state model TPR and FPR are obtained as averaged values over the two states.} 
\label{fig:ROCs}
\end{center}
\end{figure}

\FloatBarrier

\section{Empirical Application}\label{sec:app}
\subsection{Data Description}\label{subsec:data}
Over the last decades, hedge fund managers and regulators highlighted the need to accurately assess contagion and systemic risk, depending on the situation of the economy. 
In this paper, we are interested in investigating how the degree of network connectivity across different asset classes changes in severe market turbulence, which can possibly threaten the integrity of the financial system and represent a potential source of financial instability. For this purpose, we apply the proposed methodology to daily returns of $d = 29$ financial assets comprising stock market indices, commodity futures and digital currencies. The set of considered assets includes the 10 largest cryptocurrencies in terms of market capitalization, including Bitcoin, Ethereum, Ripple and Stellar, the 10 most exchanged metal and energy commodities, such as Gold, Silver, Crude Oil and Natural Gas, and the 9 major world stock indexes, among whom we include the S\&P 500, Ftse Mib, Nikkei 225 and the Shangai Composite Index.

The sample dataset is collected from Yahoo Finance and the study period starts on November 10, 2017 and ends on September 28, 2023, for a total of $T = 1250$ observations after removing missing data. The considered timespan is marked by numerous crises that may have impacted cross-market correlation patterns, spanning from the cryptocurrency bubble crisis in 2017-2018 to the global crisis sparked by the COVID-19 pandemic in 2020-2021, which have caused unprecedented levels of uncertainty and risk. 
 Daily returns with continuous compounding are calculated taking the logarithm of the difference between closing prices in consecutive trading days. In Table \ref{tab:stats} we report the list of examined variables and the summary statistics for the whole sample. As one can see, each asset displays the typical stylized facts that characterize the financial markets,
 and cryptocurrencies are generally much more volatile than commodity futures and stock market indices, having the highest standard deviations. Further, the Jarque-Bera test strongly rejects the normality assumption of daily returns for all series considered. Such findings reflect the large price variations, occurred during the period under consideration, which were triggered by the burst of the cryptocurrency bubble in 2018 and by the COVID-19 outbreak at the dawn of 2020. In concluding, the Augmented Dickey-Fuller (ADF) test (\citealt{dickey1979distribution}) shows that all daily returns are stationary at the $1\%$ level of significance.
 Figure \ref{fig:prices} shows the daily prices of
the 29 assets over the entire period of observation. At first glance we immediately recognize the high volatility that characterizes the market of the cryptocurrencies and the waves of exponential price increases. The first wave started at the end of 2018, the first big \qmo bubble\qmcsp that brought BTC above 20 thousands dollars, another one in mid-2019, and the most recent and bigger one started in late 2020, with a development over the entire course of 2021. As regards commodities and stock indexes, we observe quite a different path with respect to the digital currencies. The collapse of the financial markets at the beginning of 2020 caused by the COVID-19 pandemic represents a watershed for volatility regimes. After March 2020 we see an increase in price volatility for a large proportion of commodities and stock indexes. We also observe price peaks of Gasoline, Crude Oil and Brent during the course of 2022 caused by the sanctions imposed against Russia in response to the Ukrainian war.

\begin{figure}[!t]
\centering
\includegraphics[width=1\linewidth]{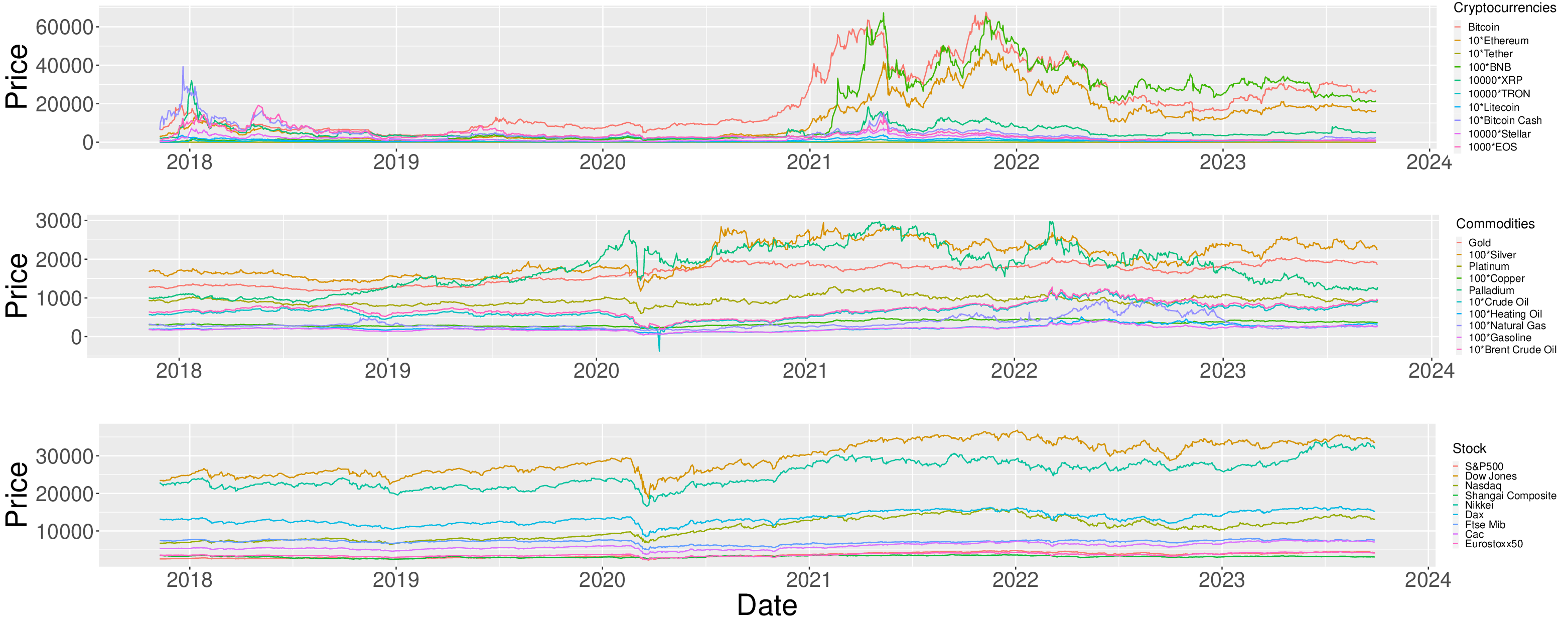}
\caption{From top to bottom daily time-series of prices of the assets considered. Observations span from November 10, 2017 to September 28, 2023.} 
\label{fig:prices}
\end{figure}

 \FloatBarrier

\begin{table}[h]
  \centering
  \scalebox{0.65}{
\begin{tabular}{rlrrrrrrrr}
\multicolumn{1}{l}{Sector} & Variable & Minimum & Mean  & Maximum & Stdev & Skewness & Kurtosis & Jarque-Bera test & ADF test \\
\midrule
      & 1: Bitcoin & -46.47 & 0.11  & 22.51 & 5.03  & -0.76 & 9.01  & \textbf{4343.7} & \textbf{-10.02} \\
      & 2: Ethereum & -55.07 & 0.13  & 32.5  & 6.47  & -0.66 & 7.73  & \textbf{3203.26} & \textbf{-10.09} \\
      & 3: Tether & -5.26 & 0     & 5.66  & 0.48  & 0.77  & 49.86 & \textbf{129594.27} & \textbf{-14.98} \\
      & 4: BNB & -54.31 & 0.38  & 84.62 & 7.46  & 1.19  & 22.42 & \textbf{26474.64} & \textbf{-9.19} \\
\multicolumn{1}{l}{Cryptocurrency} & 5: XRP USD & -69.26 & 0.07  & 74.65 & 8.21  & 1.13  & 20.88 & \textbf{22976.08} & \textbf{-9.9} \\
      & 6: TRON & -69.21 & 0.29  & 174.43 & 9.7   & 5.14  & 94.96 & \textbf{475117.14} & \textbf{-8.37} \\
      & 7: Litecoin & -44.91 & 0     & 53.98 & 6.88  & 0.18  & 8.55  & \textbf{3811.58} & \textbf{-10.22} \\
      & 8: Bitcoin Cash & -56.13 & -0.08 & 50.38 & 8.01  & 0.21  & 9.13  & \textbf{4354.43} & \textbf{-10.66} \\
      & 9: Stellar & -42.96 & 0.08  & 80.74 & 8.06  & 1.9   & 20.88 & \textbf{23461.62} & \textbf{-9.94} \\
      & 10: EOS & -50.42 & -0.06 & 65.12 & 8.18  & 0.33  & 8.78  & \textbf{4041.24} & \textbf{-10.37} \\
\midrule
      & 11: Gold & -5.11 & 0.03  & 5.81  & 1.01  & -0.22 & 4.68  & \textbf{1150.36} & \textbf{-10.5} \\
      & 12: Silver & -16.08 & 0.03  & 14.71 & 2.02  & -0.38 & 9.08  & \textbf{4324.41} & \textbf{-10.25} \\
      & 13: Platinum & -12.35 & 0     & 11.18 & 1.99  & -0.27 & 4.07  & \textbf{876.45} & \textbf{-11.32} \\
      & 14: Copper & -8.14 & 0.02  & 6.91  & 1.55  & -0.3  & 2.26  & \textbf{283.32} & \textbf{-10.93} \\
\multicolumn{1}{l}{Commodities} & 15: Palladium & -23.4 & 0.03  & 22.6  & 2.73  & -0.43 & 10.5  & \textbf{5782.71} & \textbf{-12.38} \\
      & 16: Crude Oil & -28.22 & 0.09  & 42.31 & 3.57  & 0.99  & 33.4  & \textbf{58294.49} & \textbf{-8.87} \\
      & 17: Heating Oil & -26.14 & 0.06  & 15.42 & 2.91  & -1.21 & 13.78 & \textbf{10200.87} & \textbf{-11.36} \\
      & 18: Natural Gas & -30.05 & -0.01 & 38.17 & 4.34  & 0.17  & 7.92  & \textbf{3270.09} & \textbf{-9.74} \\
      & 19: RBOB Gasoline & -50.89 & 0.05  & 28.88 & 3.46  & -2.51 & 48.5  & \textbf{123801.84} & \textbf{-9.23} \\
      & 20: Brent Crude Oil & -27.58 & 0.06  & 27.42 & 2.88  & 0.09  & 18.14 & \textbf{17149.08} & \textbf{-9.81} \\
\midrule
      & 21: S\&P500 & -12.77 & 0.04  & 8.97  & 1.38  & -0.94 & 13.81 & \textbf{10115.35} & \textbf{-10.1} \\
      & 22: Dow Jones & -13.84 & 0.03  & 10.76 & 1.38  & -1.12 & 18.85 & \textbf{18773.32} & \textbf{-10.4} \\
      & 23: Nasdaq & -13.15 & 0.06  & 8.93  & 1.62  & -0.62 & 7.07  & \textbf{2682.89} & \textbf{-10.1} \\
      & 24: Shangai Composite Index & -8.04 & -0.01 & 7.55  & 1.17  & -0.39 & 4.98  & \textbf{1324.96} & \textbf{-11.33} \\
\multicolumn{1}{l}{Stock Indexes} & 25: Nikkei & -6.27 & 0.03  & 7.73  & 1.27  & -0.01 & 3.48  & \textbf{630.63} & \textbf{-11.02} \\
      & 26: Dax & -13.05 & 0.01  & 10.41 & 1.37  & -0.56 & 12.89 & \textbf{8715.23} & \textbf{-10.25} \\
      & 27: Ftse Mib & -11.51 & 0     & 8.67  & 1.15  & -1.05 & 14.21 & \textbf{10741.41} & \textbf{-10.59} \\
      & 28: Cac & -13.1 & 0.02  & 8.06  & 1.35  & -0.92 & 12.45 & \textbf{8254.22} & \textbf{-10.56} \\
      & 29: Eurostoxx50 & -13.24 & 0.01  & 8.83  & 1.35  & -0.88 & 12.75 & \textbf{8630.23} & \textbf{-10.45} \\
\bottomrule
\end{tabular}%
    }
        \caption{Summary statistics of daily log-returns over the entire period. The Jarque-Bera and ADF test statistic is displayed in boldface when the null hypothesis is rejected at the 1\% significance level.}
  \label{tab:stats}%
\end{table}%
  \FloatBarrier

\subsection{Results}\label{subsec:results} 
Following these considerations, 
 the proposed HMGHGM can provide insights into the temporal evolution of the conditional correlations of asset returns, where the assets' multivariate distribution is not constrained to be Normal but can assume any fat-tailed distribution within the GH family.
  \\
As a first step of the empirical analysis, in order to select the optimal $\rho$ in \eqref{eq:completelpen} and number of latent states $K$, we fit the proposed HMGHGM for a sequence of 300 values of $\rho$ for $K = \{1,\dots,4\}$. 
 We also include a state-specific penalty term by setting the weight $\nu_k$ in \eqref{eq:completelpen} as the (scaled) effective sample size of state $k$, i.e., $\nu_k = \sum_{t=1}^T \gamma_t (k) / T$, where $\gamma_t (k)$ is defined in \eqref{gamma}. To select the best pair $(K, \rho)$, we consider two model selection criteria, namely a Bayesian Information Criterion (BIC) type index and the Mixture Minimum Description Length (MMDL, \citealt{figueiredo1999fitting}) defined, respectively, as

\be\label{eq:BIC}
\mbox{BIC}_{K, \rho} = -\ell(\bs \theta| \bs y_1,\dots,\bs y_T) + \frac{1}{2}\log(T)K(K-1) + \frac{1}{2}\log(T)\sum_{k=1}^K Df(k, \rho)
\ee
\be\label{eq:MMDL}
\mbox{MMDL}_{K, \rho} = -\ell(\bs \theta| \bs y_1,\dots,\bs y_T) + \frac{1}{2}\log(T)K(K-1) + \sum_{k=1}^K \frac{1}{2}\log(T \nu_k) Df(k, \rho),
\ee
where $\ell(\bs \theta| \bs y_1,\dots,\bs y_T)$ denotes the observed log-likelihood and we set the degrees of freedom as $Df(k, \rho) = d + \sum_{l \geq l'} \bs 1_{\bs \Theta_{k,\rho}^{ll'} \neq 0}$.
 The two criteria are identical, with the only difference that the MMDL adjusts the penalty term for the effective sample size $T \nu_k$ of each state $k$. 
 The optimal $K$ and $\rho$ are obtained by selecting the model with the lowest criterion value. 
To compare models with differing number of states, Table \ref{tab:criteria} reports the BIC and MMDL values for each $K$ at the optimal $\rho$. As can be seen, the BIC selects 2 states, while MMDL chooses $K = 3$ at the optimal $\rho$. Following the works of \cite{figueiredo1999fitting} and \cite{stadler2013penalized}, demonstrating that, in general, the MMDL outperforms the BIC, we thus select the proposed HMGHGM with $K=3$ states for our analyses.

\begin{table}[htbp]
  \centering
\begin{tabular}{lrrrr}
      & $K = 1$ & $K = 2$ & $K = 3$ & $K = 4$ \\
\midrule
BIC & 148657.1 & \textbf{141338.4} & 141663.0 & 143774.4 \\
MMDL & 148657.1 & 140635.4 & \textbf{140090.8} & 141096.8 \\
\bottomrule
\end{tabular}%
    \caption{BIC and MMDL values corresponding to the optimal $\rho$ with varying number of states. Bold font highlights the best values for the considered criteria (lower-is-better).}
  \label{tab:criteria}%
\end{table}%

\FloatBarrier

In light of these comments, we begin commenting on the results by the hidden process. The top of the Figure \ref{fig:posteriorprobs} represents the estimated posterior probabilities of being in latent state $j$, at time $t$, where $j = 1,...,K$ and $t = 1,...,T$, conditional on the observed time-series and with over-imposed trend according to a smoothed local regression. These probabilities allow us to make inference on the latent process and to describe the different market phases of the assets considered. 
 The bottom Figure \ref{fig:posteriorprobs} shows the decoded states obtained through local decoding by considering the maximum of the posterior probabilities and the Viterbi algorithm. The two methods convey almost identical results. The predicted trajectories indicate that the three states are visited the $38\%, 39\%$ and $23\%$ of the entire period, respectively.
 Overall, the following conclusions can be drawn. The distribution of states closely mirrors the dynamics of the considered portfolio during the last five years. In particular, State 3 manages to capture the speculative waves of the cryptocurrency market at the end of 2017, in mid-2019, and at the most recent surge during 2021. As can be observed in Figure \ref{fig:prices}, these periods are characterized by extreme price movements with exponential behaviors and can represent subsequent dramatic losses for the investors. State 1 and 2 instead succeed in mimic the behavior of commodities and stock markets during the considered time frame. State 1, which identifies the period that goes from the end of 2017 to the collapse of financial markets of March 2020 induced by the COVID-19 pandemic, is related to a phase of low prices. State 2 instead indicates the phases of the markets that, after the second half of 2020, are characterized by rise in prices. Overall, Figure \ref{fig:posteriorprobs} suggests us that the proposed model efficiently manages to identify high and low price and volatility trends that alternate over the years.\\

In order to summarize the interconnectedness among the considered assets during different market conditions, we build the regime-specific graphs $\widehat{G}_k = (V, \widehat{E}_k)$, $k=1,\dots,K$, where the variables in Table \ref{tab:stats} represent the vertices in $V$, identified by the corresponding estimated precision matrix $\widehat{\bs \Theta}_k$, $k=1,\dots,K$. Specifically, an edge between two nodes is created in the $k$-th graph if the estimated $\widehat{\Theta}_{k,i,l} \neq 0$, i.e., $(i,l) \in \widehat{E}_k$ if and only if $\widehat{\Theta}_{k,i,l} \neq 0$, for $i,l \in V, i \neq l$. Figures \ref{sfig:graphK1}, \ref{sfig:graphK2} and \ref{sfig:graphK3} report the estimated graphs for the selected three-states HMGHGM.
 To highlight the most important variables in the network, pies in light blue indicate the degree centrality measure of each node, while the edge colors specify the sign of the corresponding interaction (green = positive, red = negative). 
 In each graph, stock market indexes are coloured in yellow, cryptocurrencies in grey while commodities are shown in red, and the vertex labels are reported in Table \ref{tab:stats}. By looking at Figure \ref{sfig:graphK3}, each graph identifies a different regime of network connectivity consistently with the remarks given above. We identify a high asymmetry in network connectivity: being State 3 the representation of the speculative waves of the cryptocurrency market, the third graph shows, as expected, a strong centrality and clustering pattern of the crypto assets. The most central nodes (BNB, TRON, Litecoin, Stellar) indeed represents some of the assets with the most \qmo explosive\qmcsp trends during the crypto bubble. The strong clustering pattern is present across the three estimated graphs. Specifically, digital currencies are strongly tied to each other and rather isolated from other assets (\citealt{dyhrberg2016bitcoin, bouri2017hedge, corbet2018exploring, giudici2021crypto}). It emerges how they might act as safe-haven assets offering protection to investors against losses to offset market risk, being uncorrelated with stocks during market up and downturns. 
 In all three networks, Bitcoin is one of the least central node among the cryptos. This is in line with the work of \cite{bouri2019co}, which explains that the behaviour of crypto-traders is not to anchor to a price level: after an exponential price appreciation in one cryptocurrency (e.g., Bitcoin), they start to look for more attractive cryptocurrencies that have a better risk-reward profile. One example is represented by Stellar (9) and Litecoin (7), which result always very central in each state of the chain.
 \\
 On the other hand, volatility of commodities and stock indexes is closely related to the uncertainty of the economic outlook.  
 As one can see, in the first two states of the chain the most central nodes are represented by stocks (in particular Dow Jones, Nasdaq, Dax and Eurostoxx50) and commodities (i.e. Platinum, Copper, Crude Oil and Gasoline), who have been particularly affected by the interruptions in global production chains over the course of the pandemic. 
 We detect strong clustering patterns for all the three sectors considered. In particular we highlight the strongest connections within metals, energy, European and US indexes, representing possible co-movements during both upward and downward trends, probably due to the fact that assets of the same sectors or indexes of the same regions have responded similarly to the economic crisis caused by the pandemic.
  Finally, we highlight the role of the Gold as a safe haven asset, being one of the most disconnected nodes in the graph for every volatility level.

  \begin{figure}[!h]
\centering
  \includegraphics[width=1\linewidth]{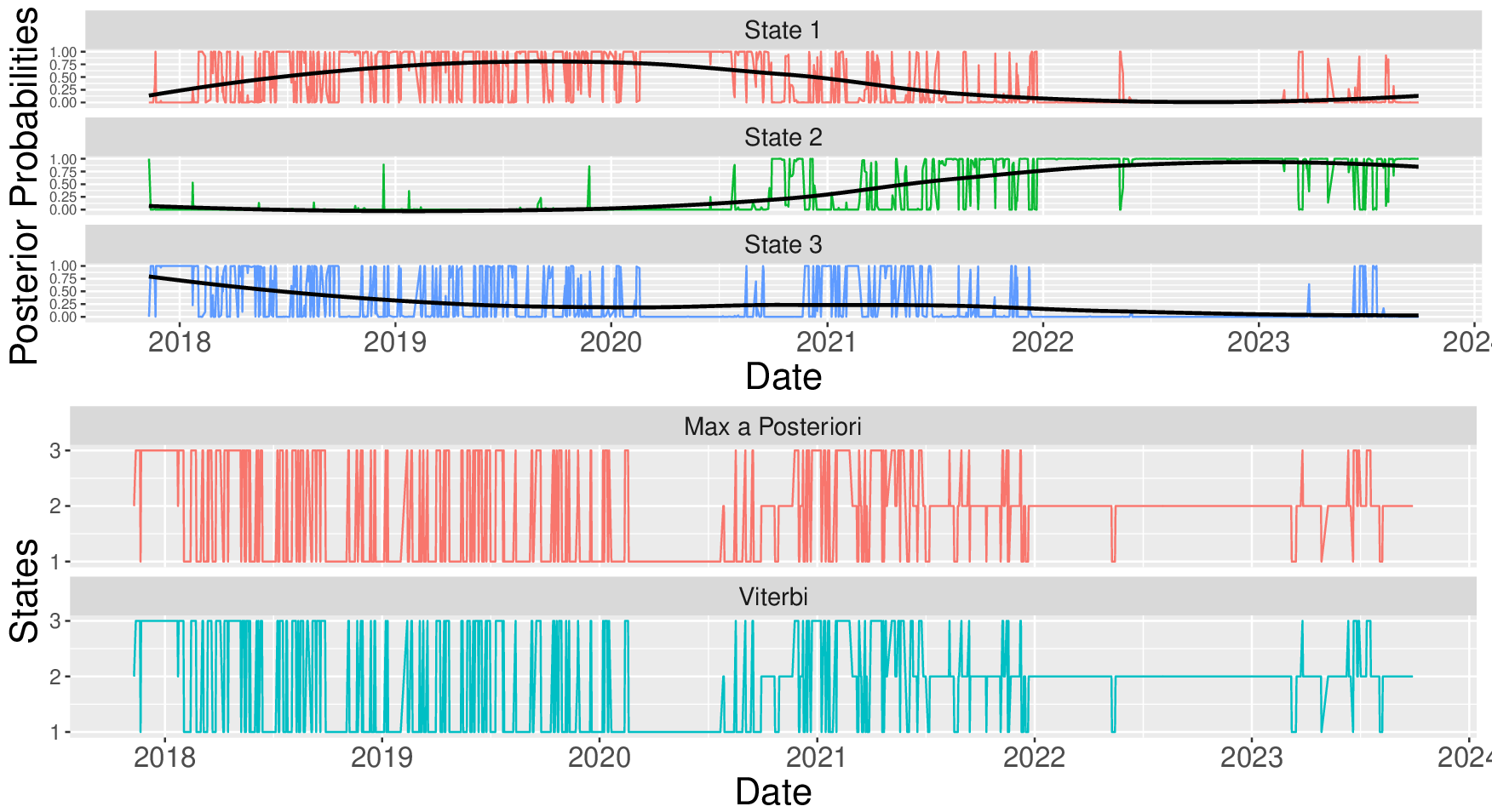}
\caption{From top to bottom, predicted posterior probabilities and predicted sequence of hidden states over time with $K = 3$ states. The over-imposed black lines denote the trend estimated by a local linear regression.}
\label{fig:posteriorprobs}
\end{figure}


 \begin{figure}[!t]
\centering
\begin{subfigure}{.5\textwidth}
  \centering
  \includegraphics[angle = 270, width=1\linewidth]{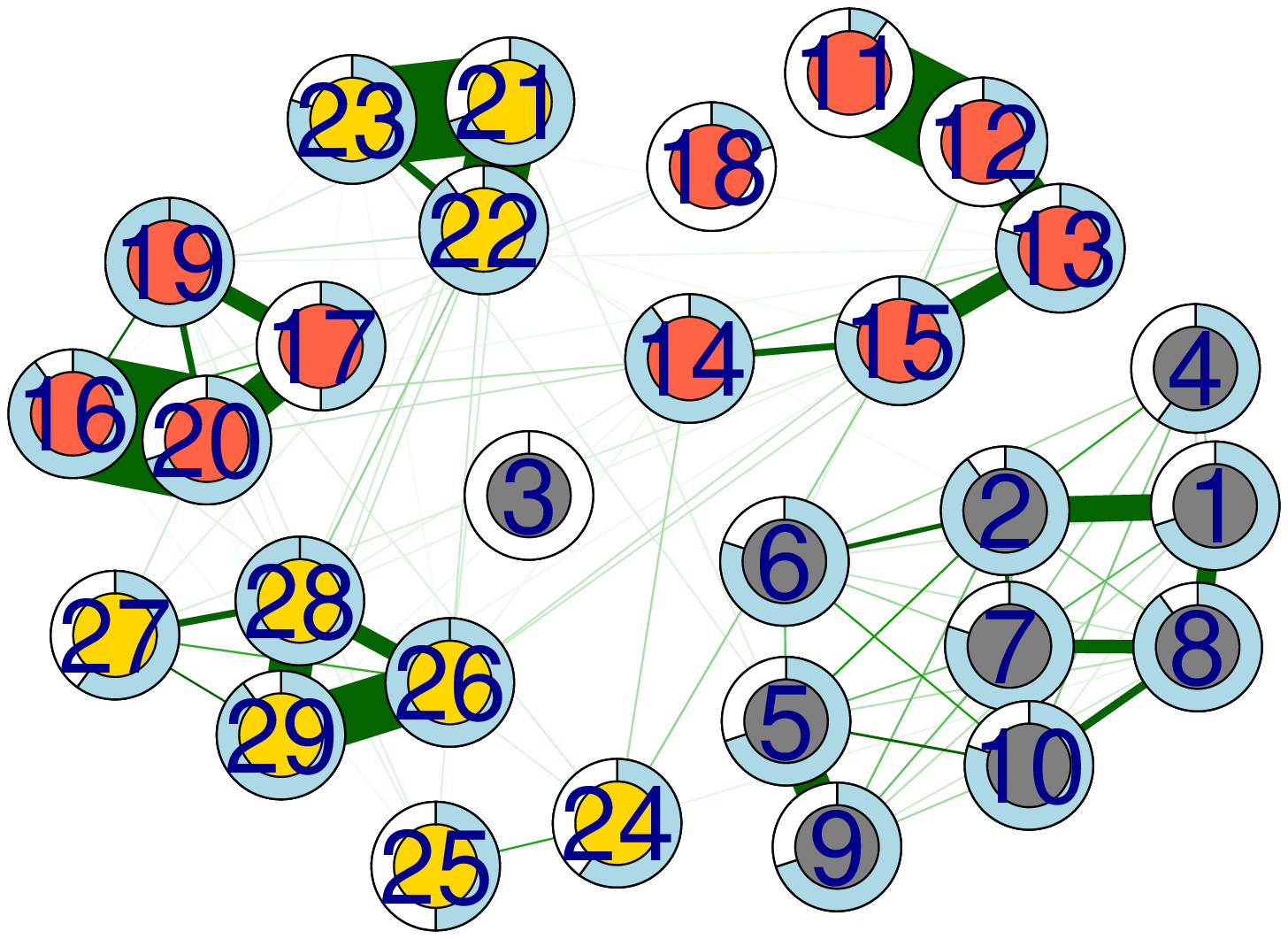}
  \caption{$\widehat{G}_{1}$}
  \label{sfig:graphK1}
\end{subfigure}%
\begin{subfigure}{.5\textwidth}
  \centering
  \includegraphics[angle = 270, width=1\linewidth]{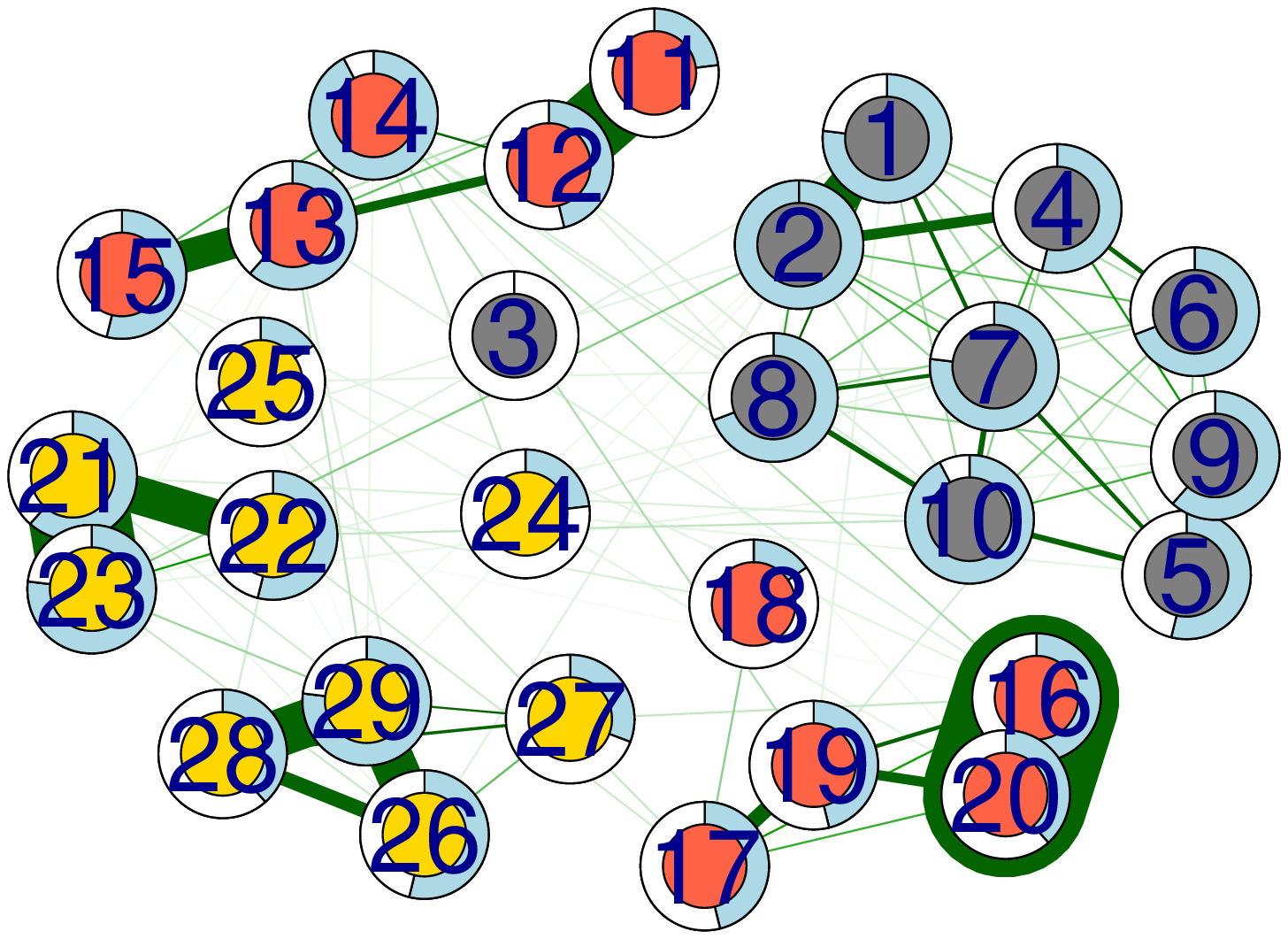}
  \caption{$\widehat{G}_{2}$}
  \label{sfig:graphK2}
\end{subfigure}\\%
\begin{subfigure}{.5\textwidth}
  \centering
  \includegraphics[angle = 270, width=1\linewidth]{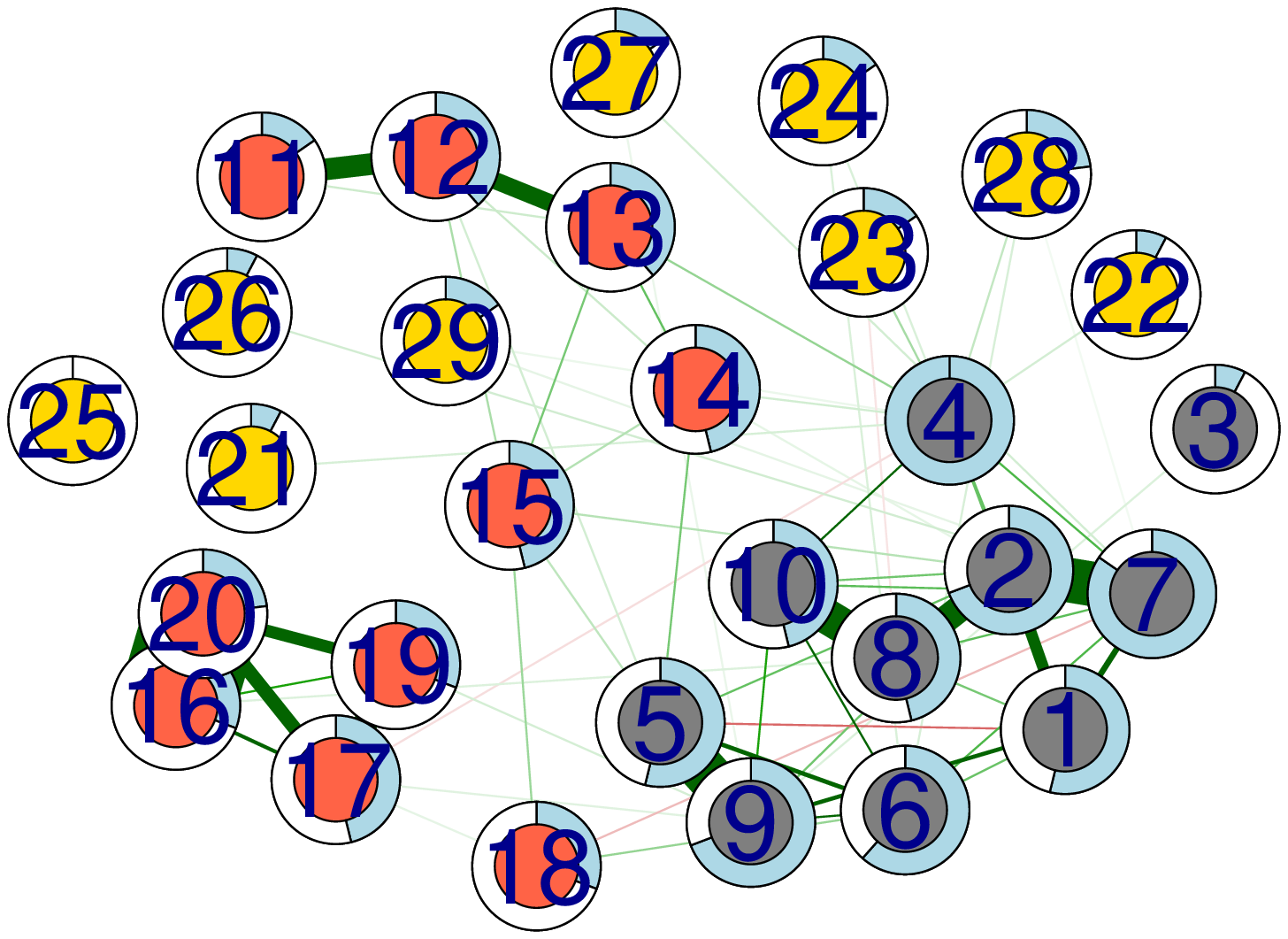}
  \caption{$\widehat{G}_{3}$}
  \label{sfig:graphK3}
\end{subfigure}
\begin{subfigure}{.35\textwidth}
  \centering
  \includegraphics[width=0.6\linewidth]{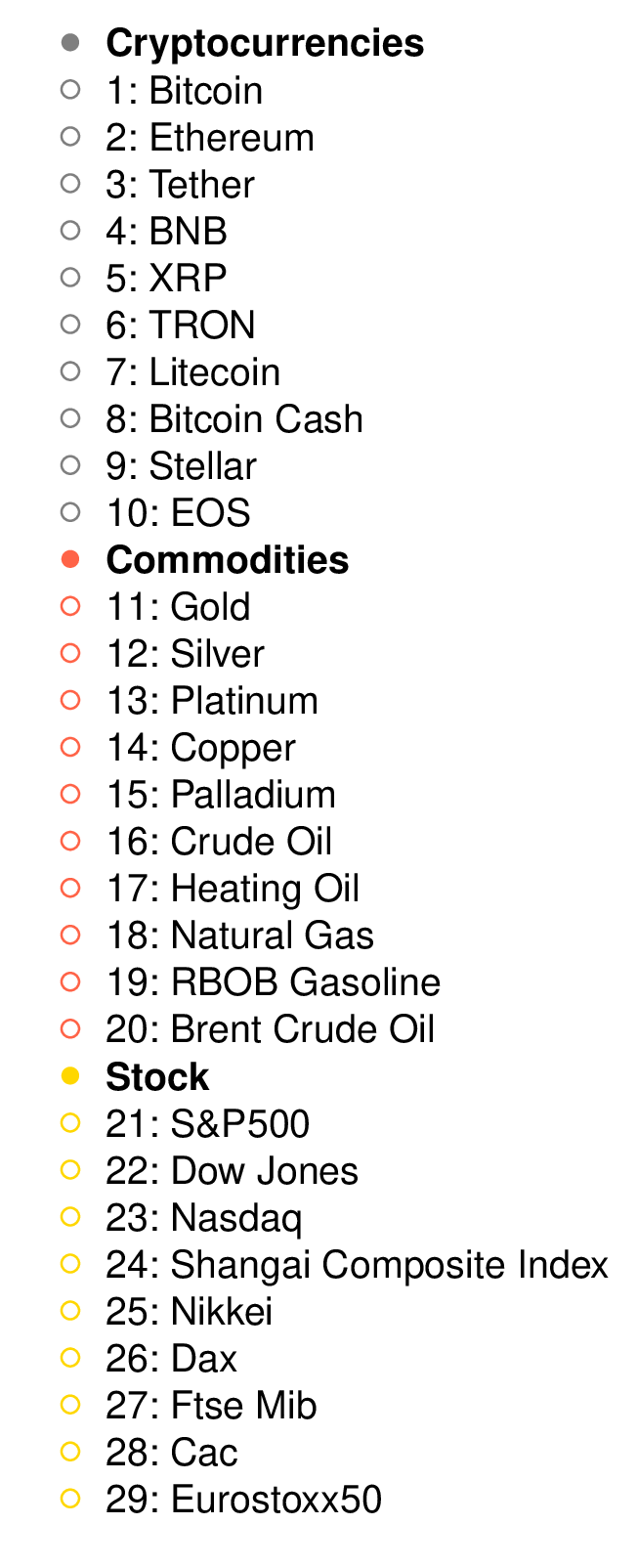}
  \caption*{}
  \label{sfig:legend}
\end{subfigure}
\caption{\small Estimated graphs for $k = 1$ (a), $k = 2$ (b) and $k = 3$ (c). Yellow, grey and red nodes represent respectively indexes, cryptocurrencies and commodities while the vertex labels are illustrated in Table \ref{tab:stats}. Green edges in the networks depict positive associations, while red edges represent negative associations. Pies in light blue indicate degree centrality.}
\label{fig:graph}
\end{figure}

\FloatBarrier

\section{Conclusions}\label{sec:concl}
This paper proposes a new sparse dynamic graphical model to study time-varying conditional correlations between important commodities, cryptocurrencies and stock indexes in a period of high volatility and market turbulence, without relying on the assumption of normally distributed data. Starting from \cite{finegold2011robust} and \cite{stadler2013penalized}, we introduce a hidden Markov graphical model where, conditional on the latent state, emission densities follow a (symmetric) GH distribution with state-dependent parameters. Since the index and concentration parameters are completely free to vary within the GH family they allow us to accommodate well a wide range of empirical characteristics of the data in real applications. To induce sparsity and recover only the most prominent relations, we develop a Penalized ECME algorithm exploiting the Gaussian location-scale mixture representation of the GH distribution with a Lasso $L_1$ penalty on the elements of the state-specific precision matrices. Our procedure is computationally efficient and can be implemented by making use of the \textit{glasso} of \cite{friedman2008sparse} within the conditional M-step of the proposed algorithm. The good performance of our model is evaluated using numerical simulations under different scenarios. 

In the real data analysis we considered daily returns of 29 cryptocurrencies, commodities and market indexes from November 2017 to September 2023 to investigate the dynamics of conditional correlation structure across different financial markets during low and high volatility periods. According to the MMDL criteria, we identified three latent regimes corresponding to different degrees of network connectivity.
 The first two regimes show negative and positive phases of the market, respectively, while the third manages to capture the peculiar explosive trends of the cryptocurrencies. The estimated graphs are consistent with the hypothesis that cryptocurrencies are highly connected to each other and disconnected from traditional asset types \citep{bouri2017hedge, corbet2018exploring, giudici2021crypto}, resulting in possible safe-haven assets during tumultuous times. Commodities that stand out in terms of degree centrality are the ones that have been the most affected by the interruptions in global production chains over the course of pandemic, namely Palladium, Copper, Crude Oil and Gasoline. 

 
 Outside the financial world, for example, a similar methodology could be employed in the field of dynamic mixed graphical models, that could convey critical information when employing not just continuous but also discrete datasets.
Future research may also extend the proposed graphical model by using multivariate asymmetric distributions within the GH family.

\FloatBarrier
\bibliographystyle{agsm}
\bibliography{biblioGH.bib}       

\end{document}